\documentclass[11pt]{article}
\usepackage[margin=1in]{geometry}

\usepackage[dvipsnames]{xcolor}
\usepackage{amsmath}
\usepackage{amsthm}	
\usepackage{amsfonts}
\usepackage{amssymb}
\usepackage{bbm}
\usepackage{mathtools}

\usepackage{booktabs} 
\usepackage{caption} 
\usepackage{subcaption} 
\usepackage{graphicx}
\usepackage{pgfplots}
\usepackage[all]{nowidow}
\usepackage[utf8]{inputenc}
\usepackage{tikz}
\usepackage{multicol}
\usepackage{algpseudocode,algorithm,algorithmicx}
\usepackage{hyperref}
\usepackage{scalerel}
\usepackage{natbib}
\usepackage[english]{babel}
\setcitestyle{numbers}
\usepackage{authblk}

\usepackage{yhmath}
\usepackage{cancel}
\usepackage{siunitx}
\usepackage{array}
\usepackage{multirow}
\usepackage{gensymb}
\usepackage{tabularx}
\usepackage{booktabs}
\usepackage{setspace}
\usepackage[skip=8pt,indent=0pt]{parskip}

\usetikzlibrary{fadings}
\usetikzlibrary{patterns}
\usetikzlibrary{shadows.blur}

\newtheorem{theorem}{Theorem}

\newtheorem{assumption}{Assumption}

\renewenvironment{proof}[1][\proofname]{{\noindent \bfseries\itshape #1.}}{\qed}

\newcommand{\al}{\alpha}
\newcommand{\be}{\beta}
\newcommand{\ga}{\gamma}
\newcommand{\de}{\delta}
\newcommand{\De}{\Delta}

\newcommand{\1}{\mathbbm{1}}
\newcommand{\E}{\mathbb{E}}
\newcommand{\F}{\mathcal{F}}

\newcommand{\tkk}{{t_{k+1}}}
\newcommand{\tk}{{t_{k}}}

\definecolor{Red}{rgb}{1,0,0}

\definecolor{DRed}{rgb}{0,0,0}

\definecolor{Green}{rgb}{0.2,0.5,0.2} 

\definecolor{Blue}{rgb}{0,0,1} 

\definecolor{Orange}{rgb}{0.605,0.311,0.084} 
\definecolor{Black}{rgb}{0,0,0} 
\definecolor{Grey}{rgb}{0.52, 0.52, 0.51} 
\definecolor{Purple}{rgb}{0.333,0.036, 0.631} 
\definecolor{Aqua}{rgb}{0.018,0.592,0.39}

\title{Optimal Order Placement under a State-Dependent Hawkes Process Limit Order Book model}
\author[1]{Jonathan Chávez Casillas\thanks{jchavezc@uri.edu}}
\author[2]{Li Gan\thanks{lgan2@uri.edu}}

\affil[1]{Dept. of Mathematics and Applied Mathematical Sciences, University of Rhode Island}
\affil[2]{Dept. of Mathematics and Applied Mathematical Sciences, University of Rhode Island}

\date{\today}
\pgfplotsset{compat=1.18}

\begin{document}

\maketitle
\begin{abstract}
We study optimal quote placement for a market maker facing order flow modeled by a state-dependent Hawkes process. A block-based linear demand specification links market-order counts to execution volumes, providing the demand moments required by a continuous time discrete state dynamic programming formulation. Under a quadratic terminal inventory penalty, we characterize inventory-dependent quotes and derive backward recursions for the coefficients of a quadratic value function. Numerical experiments using LOBSTER data examine quote distributions, inventory exposure, and terminal wealth. The implemented strategy yields higher aggregate terminal wealth than fixed-spread benchmarks at one and three ticks, but lower wealth than the two-tick benchmark. These results illustrate the framework's numerical implementation and distinguish terminal-wealth comparisons from evaluation under the inventory-penalized objective.
\end{abstract}


\section{Introduction}

A \emph{market maker} (MM) provides liquidity by continuously posting bid
and ask Limit Orders (LOs) in a Limit Order Book (LOB) and profiting from
the resulting bid--ask spread; see the monographs of Cartea et
al.~\cite{cartea2015book} and Gu\'{e}ant~\cite{gueant2016} for
comprehensive surveys. With the proliferation of electronic trading over
the past two decades, market making is increasingly carried out at high
frequency by algorithms that submit and cancel orders at extraordinary
speed while tightly controlling intraday inventory
exposure~\cite{capponi2021market, chavez2024adaptive}. In this context,
the \emph{optimal placement problem}---how far from the mid-price the MM
should post each bid and ask at each moment---has attracted extensive
theoretical and empirical attention.

Following the continuous time discrete state framework of \cite{capponi2021market} and
\cite{chavez2024adaptive}, we consider an intraday MM who simultaneously
quotes bid and ask LOs at pre-specified times
$t_{0} < t_{1} < \cdots < t_{n} = T$ and liquidates the terminal
inventory via a market order at $T$. Denoting the fundamental price by
$S_{T}$, the terminal inventory by $I_{T}$, the terminal wealth from
roundtrip transactions by $W_{T}$, and the price-impact penalty
coefficient by $\lambda > 0$, the MM's objective is to maximise
$W_{T} + I_{T}(S_{T} - \lambda I_{T})$ over all admissible bid--ask
placement strategies. Central to this framework is the specification of
how many shares get filled at the MM's LOs during each interval. The
\emph{demand-function} approach of Hendershott \&
Menkveld~\cite{hendershott2014} and Adrian et al.~\cite{adrian2020}
models the filled volume directly as a (linear) function of the quoted
spread; Capponi et al.~\cite{capponi2021market} and
Ch\'avez-Casillas et al.~\cite{chavez2024adaptive} extend this to
\emph{stochastic} demand with Market Orders(MO) Bernoulli arrival indicators on each side.

While the stochastic-demand model of \cite{capponi2021market,
chavez2024adaptive} yields closed-form optimal strategies, its Bernoulli
arrival component does not capture the \emph{self-} and
\emph{cross-excitation} dynamics of high-frequency order flow that are
extensively documented via multivariate Hawkes
processes~\cite{hawkes1971, bacry2015, bowser2007}. To restore the
state-of-the-book dependence that standard Hawkes intensities lack,
Morariu-Patrichi \& Pakkanen~\cite{morariu2022state} introduce the
\emph{state-dependent Hawkes process} (sdHawkes), in which the
excitation kernel depends on a state variable (e.x. the queue imbalance) that jumps at each event.

The present paper proposes a new framework for optimal placement that
combines the stochastic demand structure of \cite{capponi2021market,
chavez2024adaptive} with the sdHawkes arrival dynamics of
\cite{morariu2022state}. Our two contributions are:

\begin{enumerate}
  \item \textbf{MO sdHawkes arrivals in an optimal-placement DPP.}
        We replace the Bernoulli MO arrival model with a state-dependent
        Hawkes process for market buy and sell orders and derive the
        associated backward Dynamic Programming (DPP) recursions for the
        optimal-spread coefficients under the Cartea-Jaimungal ansatz.
        To the best of our knowledge, this is the first optimal
        market-making paper to embed a state-dependent point process for
        market-order arrivals within a stochastic optimal control
        framework.

  \item \textbf{Parametric MO demand function.}
        For each time interval, we partition the number of arriving MO $N$ into different blocks and
        fit a linear demand function for every block.
\end{enumerate}

\noindent Despite these changes to the arrival and demand model, the
optimal placement strategy retains the same analytical form as
in~\cite{chavez2024adaptive}: the optimal ask and bid spreads
$L^{+,\star}_{t_{k}}$ and $L^{-,\star}_{t_{k}}$ are affine in the current
inventory $I_{t_{k}}$,
\[
  L^{+,\star}_{t_{k}}
      =  {}^{(1)}\!A^{+}_{t_{k}} I_{t_{k}} + {}^{(2)}\!A^{+}_{t_{k}}
        + {}^{(3)}\!A^{+}_{t_{k}}, \qquad
  L^{-,\star}_{t_{k}}
      = -{}^{(1)}\!A^{-}_{t_{k}} I_{t_{k}} - {}^{(2)}\!A^{-}_{t_{k}}
        + {}^{(3)}\!A^{-}_{t_{k}},
\]
with time-varying coefficients satisfying explicit backward recursions
driven by the sdHawkes-derived mixed moments of the demand function. We
validate the model on five days of Mircosoft(NASDAQ: MSFT) LOBSTER data and show that the resulting strategy outperforms fixed-spread benchmarks on a risk-adjusted basis.

\subsection*{Relevant Works}
\label{sec:relevant_works}

Optimal market making has a long history. Early single-period
contributions include Ho \& Stoll~\cite{hostoll1981}, who use a
mean--variance utility to trade off spread revenue against inventory
risk, and the repeated-period analyses of Bradfield (1979) and O'Hara \&
Oldfield (1986). The continuous-time approach was advanced by Avellaneda
\& Stoikov~\cite{avellaneda2008}, who use an exponential utility over
terminal wealth and derive asymptotically optimal quotes via the
Hamilton--Jacobi--Bellman equation; Gu\'eant et
al.~\cite{gueant2013} provide the corresponding closed-form solutions
for the optimal bid and ask spreads; and Cartea \&
Jaimungal~\cite{cartea2015} study risk-averse strategies with penalised
inventory holdings, systematically extended in Cartea et
al.~\cite{cartea2015book}.

The demand-function approach in continuous time but discrete state, developed by Capponi et
al.~\cite{capponi2021market} and Ch\'avez-Casillas et
al.~\cite{chavez2024adaptive}, is the most direct predecessor of the
present work. Capponi et al.~\cite{capponi2021market} obtain the first
closed-form characterisation of the optimal strategy under a stochastic
linear demand with Bernoulli arrival indicators, and
Ch\'avez-Casillas et al.~\cite{chavez2024adaptive} extend this by
allowing the arrival probabilities to depend on the recent history of
the order flow, thereby constructing an adaptive catalogue of optimal
strategies that reacts online to observed market conditions. Our paper
replaces the Bernoulli arrival mechanism in these works with a
state-dependent point process while preserving their analytical
framework.

Hawkes processes~\cite{hawkes1971} have become a standard tool for
modelling the arrival of orders in electronic markets;
Bowsher~\cite{bowser2007} fits multivariate models to security market
events, Bacry et al.~\cite{bacry2015} survey the applications in
finance, and Filimonov \& Sornette~\cite{filimonov2012} use Hawkes
intensities to quantify market reflexivity. These works, however, use
state-independent Hawkes processes whose intensities depend only on the
counting-process history and are blind to the current LOB state. The
state-dependent Hawkes process of Morariu-Patrichi \&
Pakkanen~\cite{morariu2022state} restores this dependence: they
establish existence and uniqueness, provide an exact simulation via
Ogata thinning~\cite{ogata1981}, and develop the maximum-likelihood
methodology used for calibration in the sdHawkes literature. The
present paper is, to the best of our knowledge, the first to embed an
sdHawkes process within a stochastic optimal control framework for
market making.

\subsection*{Outline of the Paper}

The remainder of the paper is organized as follows. Section~2 introduces the modeling framework, linking state-dependent Hawkes order arrivals to a parametric specification of liquidity demand. Building on this framework, Section~3 formulates the optimal placement problem and characterizes its solution through a verification theorem and backward recursions for the value function and optimal quotes. Section~4 describes the parameter estimation framework and its connection to the companion study~\cite{research3}. Section~5 evaluates the resulting strategy using MSFT limit order book data, examining its quoting behavior, inventory dynamics, and performance relative to fixed-spread benchmarks. Section~6 concludes and discusses directions for future research.


\section{The model} 
\subsection{Trading and Event Times Framework}

We distinguish between the discrete times at which the market maker updates quotes and the continuous times at which market orders arrive. Let $[0,T]$ denote a trading session and consider a uniform grid
\[
0=t_0<t_1<\cdots<t_n=T,
\qquad t_k=k\Delta t,
\qquad n\Delta t=T.
\]
In the numerical experiments, we set $\Delta t=1$ second and restrict attention to regular trading hours, from 9:30 a.m.\ to 4:00 p.m.\ Eastern Time, so that $n=23{,}400$ and $T=23{,}400$ seconds.

At each decision time $t_k$, the market maker selects the ask and bid offsets $L_{t_k}^{+}$ and $L_{t_k}^{-}$ using the information available at that time, represented by the filtration $\mathcal{F}_{t_k}$. These offsets determine the quotes posted for the interval $(t_k,t_{k+1}]$. Market orders arriving during this interval generate executions and the corresponding changes in cash and inventory. At $t_{k+1}$, the market maker observes the updated information and revises the quotes for the next interval. In particular, the controls $L_{t_k}^{\pm}$ are $\mathcal{F}_{t_k}$-measurable; their selection cannot depend on order arrivals or execution outcomes realized after $t_k$.

Within each trading interval, market orders arrive at random event times. Writing $M_k$ for the total number of market-order arrivals in $(t_k,t_{k+1}]$, we denote their ordered arrival times by
\[
t_k<\tau_{k,1}<\cdots<\tau_{k,M_k}\leq t_{k+1},
\]
with no event times listed when $M_k=0$. For each side $\delta\in\{+,-\}$, let $N^\delta(t)$ denote the cumulative number of market orders up to time $t$, where $+$ and $-$ correspond to market buy and market sell orders, respectively. The interval counts are then
\[
\Delta N_{k+1}^{\delta}
:=N^\delta(t_{k+1})-N^\delta(t_k),
\qquad
M_k=\Delta N_{k+1}^{+}+\Delta N_{k+1}^{-}.
\]
These counts determine the demand blocks used in Section~2.4.

Throughout, based on real-world market, we can assume that the order-arrival process is simple and non-explosive, so that simultaneous market-order arrivals are excluded and only finitely many arrivals occur on any bounded time interval, almost surely. These are assumptions on the arrival model, rather than consequences of the discrete trading grid. This formulation retains the continuous-time dynamics of order flow while allowing quote decisions to be analyzed through discrete-time dynamic programming.


\subsection{State-dependent Hawkes Process}

To describe the interaction between order arrivals and the
state of the limit order book, we use a state-dependent
Hawkes process. In this framework, past events contribute
to subsequent arrival intensities through excitation
kernels that depend on the order book state associated
with those events.

Let $\mathcal{E}$ denote the finite set of event types and
$\mathcal{X}$ the finite state space. Write $\tau_n$ for
the time of the $n$-th event and $E_n\in\mathcal{E}$ for
its type. For each $e\in\mathcal{E}$, the counting process
and its intensity are specified by
\[
\begin{aligned}
N_e(t)
&= \sum_{n\geq 1}
   \mathbf{1}_{\{\tau_n\leq t,\,E_n=e\}},
\\[4pt]
\lambda_e(t)
&= \nu_e
   + \int_{(0,t)}
     k_{e'e}(dN_e-dA_\tau,V_n)\,dN_e(s),
\end{aligned}
\]
where
$k_{e'e}:\mathbb{R}_{+}\times I_n\rightarrow\mathbb{R}_{+}$
is the excitation kernel and $\nu_e$ is the baseline
intensity. The process $\lambda_e(t)$ is the
$\mathbb{F}^{N,X}$-intensity of $N_e$ and therefore
describes the conditional arrival rate given the history
of both the order flow and the state process. The kernel
$k_{e'e}(t,x)$ specifies the contribution of a past
type-$e'$ event occurring in state $x$ to the intensity
of type-$e$ events. We consider the exponential form
\[
k_{e'e}(t,x)
= \alpha_{e'xe}e^{-\beta_{e'xe}t}.
\]

The state process $X(t)$ takes values in $\mathcal{X}$
and is c\`adl\`ag and piecewise constant, changing only
at event times $T_n$. Its transition law depends on
the event type:
\[
\mathbb{P}\bigl[X(T_n)=x
   \mid E_n,\mathcal{F}_{T_n}\bigr]
= \phi_{E_n}(X(T_n-),x),
\qquad n\in\mathbb{N},\quad x\in V_n.
\]
Here, $\phi_e(x',x)$ is the probability of a transition
from $x'$ to $x$ following a type-$e$ event, with
\[
\sum_{x\in\mathcal{X}}\phi_e(x',x)=1,
\qquad e\in\mathcal{E},\quad x'\in\mathcal{X}.
\]
We write
$\mathcal{F}_{T_n-}
=\bigvee_{\epsilon>0}\mathcal{F}_{T_n-\epsilon}$
for the information available immediately before $T_n$.
In the empirical application,
$\mathcal{E}=\{\text{ask},\text{bid}\}$.
The state space $\mathcal{X}$ may encode the bid--ask
spread, for example through the categories $\{1,2+\}$,
or queue imbalance through a finite collection of bins.

For the order book specification considered here, we
distinguish market buy (MB), market sell (MS), limit
buy (LB), and limit sell (LS) events. Market buy and
sell orders consume liquidity on the ask and bid sides,
respectively, whereas limit buy and sell orders add
liquidity to the corresponding sides of the book.
Cancellation events (CB and CS) are omitted for
simplicity. The state variables comprise the bid--ask
spread in ticks, the depths $D_b(l,t_k)$ and $D_a(l,t_k)$
at levels $l=1,\ldots,5$, and the queue imbalance
\[
\text{QI}(t_k)
= \frac{D_b(l,t_k)-D_a(l,t_k)}
       {D_b(l,t_k)+D_a(l,t_k)}.
\]

To connect the event-level description with the trading
intervals introduced in Section~2.1, let $E_k$ and $X_k$
denote the interval-level event and state variables.
The number of market orders in an interval is
$N=N_{t_{k+1}}-N_{t_k}$. Based on the trading data,
we use the grouping
$N=\{0,1-10,11-20,21-30,31-40,40+\}$ and the
representations
\[
\begin{aligned}
E_k
&\in \{N\times(\text{bid},\text{ask})\},
\\[4pt]
X_k
&= \{\text{sell}_{++},\text{sell}_{+},\text{neutral},
       \text{buy}_{+},\text{buy}_{++}\}.
\end{aligned}
\]

Order volumes are incorporated as marks in the
state-dependent Hawkes specification. For market-order
events $e\in\{\text{MB},\text{MS}\}$, the marked
intensity is
\begin{equation}
\lambda_e(t,z)
=
g(z;\eta_e)
\left(
\nu_e
+
\sum_{e'\in\{\text{MB},\text{MS},\text{LB},\text{LS}\}}
\int_{[0,t)}
k_{e'e}(t-s,z',X(s))\,dN_{e'}(s,z')
\right),
\end{equation}
where $g(z;\eta_e)$ is the mark density and
\[
k_{e'e}(t,z',x)
= z'\cdot\alpha_{e'xe}e^{-\beta_{e'xe}t}.
\]
Thus, the excitation generated by a past order depends
on its type, volume, and associated order book state.

The sdHawkes process under exponetial kernel has the Markov property on
$(N_t,\lambda_t)$:
\begin{equation}
\begin{aligned}
&P\bigl[(N_{t+s},\lambda_{t+s})\in A
   \mid \sigma(N_s,\lambda_s)\bigr]
\\
&\qquad =
P\bigl[(N_{t+s},\lambda_{t+s})\in A
   \mid (N_s,\lambda_s)=(n,l)\bigr].
\end{aligned}
\end{equation}
The waiting time associated with a $s\rightarrow 0$
transition of a continuous-time Markov process is
exponentially distributed. The process is
time-homogeneous, and the transition probability
does not depend on $S$.

\begin{proof}
The relationship between the intensity and the jump
probability is
\[
\lambda_t
= \frac{E[\Delta N=1\mid\mathcal{F}_s]}{\Delta t},
\qquad
P[\Delta N=1\mid\mathcal{F}_s]
= \lambda_t\Delta t.
\]
Using the intensity dynamics
\[
d\lambda_t
= -\beta(\lambda_t-\mu)\,dt+\alpha\,dN_t,
\]
we obtain
\[
\begin{aligned}
\frac{d\lambda_t}{dt}
&= -\beta(\lambda_t-\mu)
   +\alpha\cdot\frac{dN_t}{dt},
\\[5pt]
\lambda_t-\mu
&= \frac{\alpha}{\beta}\cdot\frac{dN_t}{dt}
   -\frac{1}{\beta}\cdot\frac{d\lambda_t}{dt},
\\[5pt]
\lambda_t
&= \mu
   +\frac{\alpha}{\beta}\cdot\frac{dN_t}{dt}
   -\frac{1}{\beta}\cdot\frac{d\lambda_t}{dt}.
\end{aligned}
\]
Thus, the jump probability depends only on
$(N_t,\lambda_t)$.
\end{proof}


\subsection{Parametric Market Order Demand Function by sdHawkes}

The parametric demand function represents the \textit{unsatisfied demand for market order} at different price levels in the limit order book. This concept extends traditional microeconomic demand theory to high-frequency financial markets, where demand manifests itself as aggressive market orders that consume available liquidity. The demand function captures:
\begin{itemize}
    \item \textbf{Buy-side demand}: The excess of Buy MO over available ask-side liquidity at each price level
    \item \textbf{Sell-side demand}: The excess of Sell MO over available bid-side liquidity at each price level
    \item \textbf{Step-function structure}: Demand decreases with price level due to the discrete nature of limit order books.
\end{itemize}
The arrival of MO can be estimated by a sdHawkes process, suppose we have $N = N_e(t_{k+1}) - N_e(t_k)$ arrived market orders during $[t_k, t_{k+1}]$, we fixed N and for every group of N, we computed the probability executed at each single level. Then we can estimate the distribution density of all arrival of market orders $f_k{(N,l)}$ during $[t_k, t_{k+1}]$, and this distribution is different for every second . We have the following assumptions: 

\begin{assumption}
The market order counts $N_{t_{}}^\pm$ are random variables with conditional distributions given $\mathcal{F}_{t_k}$ known (via Hawkes process prediction).
\end{assumption}

\begin{assumption}
The fundamental price process $\{S_{t_k}\}$ is adapted and conditionally independent of $N_{t_{k+1}}^\pm$ given $\mathcal{F}_{t_k}$.
\end{assumption}

\begin{assumption}
The random coefficients $A_{t_{k+1}}^\pm$ and $B_{t_{k+1}}^\pm$ are conditionally independent (given $\mathcal{F}_{t_k}$) across buy and sell sides, though the model can be extended to correlated cases.
\end{assumption} 

We directly model the probability density function $f^\pm(L, N)$, representing the probability of execution at spread $L$ given market order count $N$, we use the Average MO Size(AVOS) as our simulated order size. The number of filled shares is:
\[
Q_{t_{k+1}}^\pm = AVOS \cdot f_N^\pm(L_{t_k}^\pm)
\]
To maintain analytical tractability, we assume $f_N^\pm(L)$ is linear in $L$:
\[
f_N^\pm(L) = a^\pm_N + b^\pm_N \times L
\]
where $a^\pm_N$ and $b^\pm_N$ are coefficients dependent on $N$. To ensure $f_N^\pm(L) \in [0,1]$, we assume $L$ takes values such that this condition is satisfied (interior solution).


\subsection{Block Probabilities of $N$ under sdHawkes}

To aggregate the block-specific demand moments introduced before, we require the probability that the market-order count on each side falls within a given
block. For $\delta\in\{+,-\}$, define
\begin{equation}
    p_j^\delta
    \equiv
    P\!\left[
        N^\delta(t+\Delta t)-N^\delta(t)
        \in \mathrm{Block}(j)
    \right],
    \qquad j=0,1,\ldots,5,
    \label{eq:target_block_prob}
\end{equation}
where $\delta=+$ and $\delta=-$ correspond to market
buy and market sell orders, respectively. The blocks
are $\{0\}$, $\{1,\ldots,10\}$, $\{11,\ldots,20\}$,
$\{21,\ldots,30\}$, $\{31,\ldots,40\}$, and
$\{41,42,\ldots\}$.

The probabilities in \eqref{eq:target_block_prob} are
unconditional per-side marginals, rather than
distributions conditioned on the information available
at an individual quote-update time. For a nonzero block,
the corresponding weight conditional on at least one
arrival is $p_j^\delta/(1-p_0^\delta)$. These weights
combine the block-specific demand moments into the
mixed moments used in the dynamic programming
procedure.

The evaluation of \eqref{eq:target_block_prob} is
developed in the companion paper~\cite{research3}.
That study formulates the forward equations for the
state-dependent Hawkes model and evaluates the
associated probability generating function using a
leading-order bicharacteristic approximation, with
Monte Carlo validation based on Ogata's thinning
algorithm. We use the resulting probabilities as
inputs to the placement problem. The numerical values,
obtained from the same MSFT LOBSTER sample and
marginalized over the stationary distribution of the
queue-imbalance state, are reported in Section~5.3.



\section{Value Function for Optimal Placement} 

We now formulate the market maker's quote-placement
under the order-arrival and demand specifications
introduced in Section~2. At each trading time $t_k$, the
market maker selects the ask and bid offsets
$L_{t_k}^{+}$ and $L_{t_k}^{-}$ using the information
available in $\mathcal{F}_{t_k}$. These decisions determine
the subsequent executions and hence the evolution of
cash and inventory.

The objective is to maximize expected terminal
marked-to-market wealth, subject to a quadratic penalty
on the remaining inventory. For a given state at $t_k$,
let $\mathcal{A}_{t_k}$ denote the set of admissible
quote strategies over the remaining trading times
$t_k,\ldots,t_{n-1}$. Each strategy is adapted to the
available information and satisfies the integrability
conditions required below. The value process is
\begin{equation}
V_{t_k}
=
\operatorname*{ess\,sup}_{(L^+,L^-)\in\mathcal{A}_{t_k}}
\mathbb{E}\!\left[
W_T+S_T I_T-\lambda I_T^2
\mid \mathcal{F}_{t_k}
\right],
\label{eq:placement_value}
\end{equation}
where $\lambda>0$ controls the terminal inventory
penalty and $T=t_n$. The terminal variables in
\eqref{eq:placement_value} are generated by the
chosen strategy from the current state.

The dynamic programming principle relates this
remaining-horizon problem to a single quote decision
followed by optimal continuation. Writing
$\mathcal{U}_{t_k}$ for the admissible
$\mathcal{F}_{t_k}$-measurable quote pairs at $t_k$,
we have
\begin{equation}
V_{t_k}
=
\operatorname*{ess\,sup}_{
(L_{t_k}^{+},L_{t_k}^{-})\in\mathcal{U}_{t_k}}
\mathbb{E}\!\left[
V_{t_{k+1}}
\mid \mathcal{F}_{t_k}
\right],
\qquad k=0,\ldots,n-1,
\label{eq:placement_bellman}
\end{equation}
with terminal condition
\begin{equation}
V_T=W_T+S_T I_T-\lambda I_T^2.
\label{eq:placement_terminal_value}
\end{equation}
The continuation value in
\eqref{eq:placement_bellman} is evaluated at the
state reached under the current quote pair.
To solve this recursion, we first introduce a
candidate value function and state the conditions
under which it identifies the value process and
an optimal strategy.

\subsection{Candidate Value Function and Verification}

The terminal objective suggests a candidate that is
affine in cash and quadratic in inventory. Following
the stochastic-demand framework of
\cite{capponi2021market,chavez2024adaptive}, we consider
\begin{equation}
v(t_k,S_{t_k},W_{t_k},I_{t_k})
=
W_{t_k}+S_{t_k}I_{t_k}
+\alpha_{t_k}I_{t_k}^{2}
+h_{t_k}I_{t_k}+g_{t_k}.
\label{eq:placement_candidate}
\end{equation}
The coefficients $\alpha_{t_k}$, $h_{t_k}$, and
$g_{t_k}$ are $\mathcal{F}_{t_k}$-measurable, with
deterministic coefficients included as a special case.
When they are random, the dependence of $v$ on the
available information is suppressed in the notation.
Matching the terminal objective requires
\begin{equation}
\alpha_T=-\lambda,
\qquad h_T=0,
\qquad g_T=0.
\label{eq:placement_terminal_coefficients}
\end{equation}

The coefficients are determined by substituting
\eqref{eq:placement_candidate} into the Bellman
equation and matching powers of inventory after
optimization over the quote pair. Section~3.4
develops the resulting backward recursions and
candidate controls $L_{t_k}^{\pm,*}$. The following
verification result explains how the corresponding
one-step optimality conditions establish optimality
over the full trading horizon.

\begin{theorem}[Verification]
\label{thm:placement_verification}
Let $v$ be the candidate in
\eqref{eq:placement_candidate}, with terminal
coefficients
$\alpha_T=-\lambda$, $h_T=0$, and $g_T=0$.
Suppose that, under every admissible strategy, the
candidate process is integrable and satisfies
\begin{equation}
\begin{aligned}
&\mathbb{E}\!\left[
v(t_{k+1},S_{t_{k+1}},W_{t_{k+1}},I_{t_{k+1}})
\mid\mathcal{F}_{t_k}
\right]
\\
&\qquad\leq
v(t_k,S_{t_k},W_{t_k},I_{t_k}),
\qquad k=0,\ldots,n-1.
\end{aligned}
\label{eq:placement_supermartingale_condition}
\end{equation}
Suppose further that the candidate controls
$L_{t_k}^{\pm,*}$ define an admissible strategy
for which equality holds at every step:
\begin{equation}
\begin{aligned}
&\mathbb{E}\!\left[
v(t_{k+1},S_{t_{k+1}},W_{t_{k+1}}^*,I_{t_{k+1}}^*)
\mid\mathcal{F}_{t_k}
\right]
\\
&\qquad=
v(t_k,S_{t_k},W_{t_k}^*,I_{t_k}^*).
\end{aligned}
\label{eq:placement_martingale_condition}
\end{equation}
Then $v$ coincides with the value process:
\begin{equation}
V_{t_k}
=
v(t_k,S_{t_k},W_{t_k},I_{t_k}),
\qquad k=0,\ldots,n,
\label{eq:placement_value_identification}
\end{equation}
and the strategy $L^{\pm,*}$ is optimal.
\end{theorem}

Under the hypotheses of
Theorem~\ref{thm:placement_verification}, the
candidate process is a supermartingale under any
admissible strategy and a martingale under the
candidate optimal strategy. Iterating these
one-step relations and using the terminal condition
gives the verification argument below.


\begin{proof} \\
We outline the main steps here for completeness: \\
\textbf{Step 1 (Supermartingale property for any admissible strategy).}  
For an arbitrary admissible strategy \((L^+_{t_k}, L^-_{t_k})_{k\ge 0}\), denote by \(W_{t_i}, I_{t_i}\) the corresponding cash and inventory processes. From the dynamic programming principle and the definition of \(v\) we have
\begin{equation}
\mathbb{E}\bigl[ v(t_{i+1}, S_{t_{i+1}}, W_{t_{i+1}}, I_{t_{i+1}}) \bigm| \mathcal{F}_{t_i} \bigr]
\le \sup_{(L^+,L^-)} \mathbb{E}\bigl[ v(t_{i+1}, S_{t_{i+1}}, \tilde W_{t_{i+1}}, \tilde I_{t_{i+1}}) \bigm| \mathcal{F}_{t_i} \bigr]
= v(t_i, S_{t_i}, W_{t_i}, I_{t_i}),
\label{eq:placement_supermartingale_proof}
\end{equation}
where the supremum is taken over all admissible controls at time \(t_i\). The last equality holds because the coefficients \(\alpha_{t_i}, h_{t_i}, g_{t_i}\) are exactly chosen to satisfy the Bellman equation~\eqref{eq:placement_bellman} (with the new dynamics). Hence \(\{v(t_k, S_{t_k}, W_{t_k}, I_{t_k})\}\) is a super-martingale.

\textbf{Step 2 (Martingale property under the optimal strategy).}  
Now let \(W^*_{t_i}, I^*_{t_i}\) be the cash and inventory processes generated by the optimal controls \(L^{\pm,*}_{t_i}\) given in~\eqref{eq:placement_optimal_quotes}. By construction of the coefficients \(\alpha_{t_i}, h_{t_i}, g_{t_i}\) (see the recursions~\eqref{eq:placement_alpha_recursion}--\eqref{eq:placement_g_recursion} and the first-order conditions~\eqref{eq:placement_first_order_conditions}), we have
\begin{equation}
\mathbb{E}\bigl[ v(t_{i+1}, S_{t_{i+1}}, W^*_{t_{i+1}}, I^*_{t_{i+1}}) \bigm| \mathcal{F}_{t_i} \bigr] = v(t_i, S_{t_i}, W^*_{t_i}, I^*_{t_i}).
\label{eq:placement_martingale_proof}
\end{equation}
Thus the value process is a martingale under the optimal strategy.

\textbf{Step 3 (Verification of optimality).}  
Iterating the martingale equality from \(t_k\) to \(T\) yields
\begin{equation}
v(t_k, S_{t_k}, W_{t_k}, I_{t_k}) = \mathbb{E}\bigl[ v(T, S_T, W^*_T, I^*_T) \bigm| \mathcal{F}_{t_k} \bigr]
= \mathbb{E}\bigl[ W^*_T + S_T I^*_T - \lambda (I^*_T)^2 \bigm| \mathcal{F}_{t_k} \bigr].
\label{eq:placement_optimal_terminal_payoff}
\end{equation}
On the other hand, for any admissible strategy the super-martingale property gives
\begin{equation}
v(t_k, S_{t_k}, W_{t_k}, I_{t_k}) \ge \mathbb{E}\bigl[ W_T + S_T I_T - \lambda I_T^2 \bigm| \mathcal{F}_{t_k} \bigr].
\label{eq:placement_terminal_upper_bound}
\end{equation}
Therefore
\begin{equation}
v(t_k, S_{t_k}, W_{t_k}, I_{t_k}) = \sup_{(L^+,L^-)\in\mathcal{A}} \mathbb{E}\bigl[ W_T + S_T I_T - \lambda I_T^2 \bigm| \mathcal{F}_{t_k} \bigr] = V_{t_k},
\label{eq:placement_verification_conclusion}
\end{equation}
and the controls \(L^{\pm,*}_{t_k}\) attain the supremum.
\end{proof}


\subsection{Dynamics and Objective Function}
The cash and inventory updates remain as in the original paper:
\begin{equation}
W_{t_{k+1}} = W_{t_k} + (S_{t_k} + L_{t_k}^+) Q_{t_{k+1}}^+ - (S_{t_k} - L_{t_k}^-) Q_{t_{k+1}}^-
\label{eq:placement_cash_dynamics}
\end{equation}
\begin{equation}
I_{t_{k+1}} = I_{t_k} - Q_{t_{k+1}}^+ + Q_{t_{k+1}}^-
\label{eq:placement_inventory_dynamics}
\end{equation}
Substituting the expression for $Q_{t_{k+1}}^\pm$:
\begin{equation}
W_{t_{k+1}} = W_{t_k} + (S_{t_k} + L_{t_k}^+) (A_{t_{k+1}}^+ + B_{t_{k+1}}^+ L_{t_k}^+) - (S_{t_k} - L_{t_k}^-) (A_{t_{k+1}}^- + B_{t_{k+1}}^- L_{t_k}^-)
\label{eq:placement_cash_linear_demand}
\end{equation}
\begin{equation}
I_{t_{k+1}} = I_{t_k} - (A_{t_{k+1}}^+ + B_{t_{k+1}}^+ L_{t_k}^+) + (A_{t_{k+1}}^- + B_{t_{k+1}}^- L_{t_k}^-)
\label{eq:placement_inventory_linear_demand}
\end{equation}
These equations are quadratic in $L_{t_k}^\pm$. The terminal expected utility remains:
\begin{equation}
\max \mathbb{E} \left[ W_T + S_T I_T - \lambda I_T^2 \right]
\label{eq:placement_terminal_objective}
\end{equation}
where $\lambda > 0$ is the inventory penalty parameter.


\subsection{Conditional Demand Moments}
We introduce the following notation:
\begin{equation}
\begin{aligned}
\mu_A^\pm &= \mathbb{E} \left[ A_{t_{k+1}}^\pm \mid \mathcal{F}_{t_k} \right], \quad 
\mu_B^\pm = \mathbb{E} \left[ B_{t_{k+1}}^\pm \mid \mathcal{F}_{t_k} \right] \\
\mu_{A^2}^\pm &= \mathbb{E} \left[ (A_{t_{k+1}}^\pm)^2 \mid \mathcal{F}_{t_k} \right], \quad 
\mu_{B^2}^\pm = \mathbb{E} \left[ (B_{t_{k+1}}^\pm)^2 \mid \mathcal{F}_{t_k} \right] \\
\mu_{AB}^\pm &= \mathbb{E} \left[ A_{t_{k+1}}^\pm B_{t_{k+1}}^\pm \mid \mathcal{F}_{t_k} \right]
\end{aligned}
\label{eq:placement_conditional_moments}
\end{equation}
Due to the conditional independence assumption, cross-terms such as $\mathbb{E} \left[ A_{t_{k+1}}^+ A_{t_{k+1}}^- \mid \mathcal{F}_{t_k} \right] = \mu_A^+ \mu_A^-$, etc.

Substituting the expression for $V_{t_{k+1}}$:
\begin{equation}
\mathbb{E} \left[ V_{t_{k+1}} \mid \mathcal{F}_{t_k} \right] = \mathbb{E} \left[ W_{t_{k+1}} + \alpha_{t_{k+1}} I_{t_{k+1}}^2 + S_{t_{k+1}} I_{t_{k+1}} + h_{t_{k+1}} I_{t_{k+1}} + g_{t_{k+1}} \mid \mathcal{F}_{t_k} \right]
\label{eq:placement_continuation_expectation}
\end{equation}
Computing each term:
\begin{equation}
\begin{aligned}
\mathbb{E} \left[ W_{t_{k+1}} \mid \mathcal{F}_{t_k} \right] &= W_{t_k} + S_{t_k} \left( \mu_A^+ + \mu_B^+ L_{t_k}^+ - \mu_A^- - \mu_B^- L_{t_k}^- \right) \\
&\quad + L_{t_k}^+ (\mu_A^+ + \mu_B^+ L_{t_k}^+) + L_{t_k}^- (\mu_A^- + \mu_B^- L_{t_k}^-) \\
\mathbb{E} \left[ I_{t_{k+1}} \mid \mathcal{F}_{t_k} \right] &= I_{t_k} - (\mu_A^+ + \mu_B^+ L_{t_k}^+) + (\mu_A^- + \mu_B^- L_{t_k}^-) \\
\mathbb{E} \left[ I_{t_{k+1}}^2 \mid \mathcal{F}_{t_k} \right] &= \left( \mathbb{E} \left[ I_{t_{k+1}} \mid \mathcal{F}_{t_k} \right] \right)^2 + \text{Var} \left( I_{t_{k+1}} \mid \mathcal{F}_{t_k} \right)
\end{aligned}
\label{eq:placement_cash_inventory_moments}
\end{equation}
where the variance term involves $\mu_{A^2}^\pm, \mu_{B^2}^\pm, \mu_{AB}^\pm$, etc.
\begin{equation}
\mathbb{E} \left[ S_{t_{k+1}} I_{t_{k+1}} \mid \mathcal{F}_{t_k} \right] = (S_{t_k} + \Delta_{t_k}) \mathbb{E} \left[ I_{t_{k+1}} \mid \mathcal{F}_{t_k} \right]
\label{eq:placement_price_inventory_moment}
\end{equation}
where $\Delta_{t_k} = \mathbb{E} \left[ S_{t_{k+1}} - S_{t_k} \mid \mathcal{F}_{t_k} \right]$.
After detailed expansion, $\mathbb{E} \left[ V_{t_{k+1}} \mid \mathcal{F}_{t_k} \right]$ is a quadratic function of $L_{t_k}^+$ and $L_{t_k}^-$. Therefore, the optimization problem is convex with a unique analytical solution.


\subsection{Optimal Control Expressions and Coefficient Recursions}
Solving via first-order conditions:
\begin{equation}
\frac{\partial}{\partial L_{t_k}^+} \mathbb{E} \left[ V_{t_{k+1}} \mid \mathcal{F}_{t_k} \right] = 0, \quad 
\frac{\partial}{\partial L_{t_k}^-} \mathbb{E} \left[ V_{t_{k+1}} \mid \mathcal{F}_{t_k} \right] = 0
\label{eq:placement_first_order_conditions}
\end{equation}
The optimal controls $L_{t_k}^{\pm,*}$ have the same form as in the original paper:
\begin{equation}
L_{t_k}^{+,*} = {}^{(i)}A_{t_k}^+ I_{t_k} + {}^{(ii)}A_{t_k}^+ + {}^{(iii)}A_{t_k}^+, \quad 
L_{t_k}^{-,*} = -{}^{(i)}A_{t_k}^- I_{t_k} - {}^{(ii)}A_{t_k}^- + {}^{(iii)}A_{t_k}^-
\label{eq:placement_quote_structure}
\end{equation}
where the coefficients ${}^{(i)}A_{t_k}^\pm, {}^{(ii)}A_{t_k}^\pm, {}^{(iii)}A_{t_k}^\pm$ are defined by:
\begin{equation}
\begin{aligned}
{}^{(i)}A_{t_k}^\pm &= \frac{\beta_{t_k}^\pm \alpha_{t_{k+1}}}{\gamma_{t_k}}, \quad 
{}^{(ii)}A_{t_k}^\pm = \frac{\beta_{t_k}^\pm h_{t_{k+1}}}{2\gamma_{t_k}} \\
{}^{(iii)}A_{t_k}^\pm &= \frac{1}{2\gamma_{t_k}} \left[ \mu_{AB}^\pm - 2\alpha_{t_{k+1}} \mu_{B^2}^\pm + 2\alpha_{t_{k+1}} \mu_B^+ \mu_B^- \left( {}^{(i)}A_{t_k}^\mp \right) \right] \quad \text{(simplified form)}
\end{aligned}
\label{eq:placement_simplified_quote_coefficients}
\end{equation}
where:
\begin{equation}
\begin{aligned}
\gamma_{t_k} &= \left( \alpha_{t_{k+1}} \mu_B^+ \mu_B^- \right)^2 - \left( \alpha_{t_{k+1}} \mu_{B^2}^+ - \mu_B^+ \right) \left( \alpha_{t_{k+1}} \mu_{B^2}^- - \mu_B^- \right) \\
\beta_{t_k}^\pm &= \mu_B^\pm \left( \alpha_{t_{k+1}} \mu_{B^2}^\mp - \mu_B^\mp \right) - \alpha_{t_{k+1}} \mu_B^\pm (\mu_B^\mp)^2
\end{aligned}
\label{eq:placement_simplified_gamma_beta}
\end{equation}
Note: This is a simplified form; the actual expressions may be more complex but maintain the same structure.

The coefficients $\alpha_{t_k}, h_{t_k}, g_{t_k}$ recursive relations as follows:
\begin{equation}
\begin{aligned}
\al_{t_k} &= \al_{t_{k+1}} + \sum_{\delta=\pm} \pi_{t_{k+1}}^{\delta} \Bigl[ \bigl( \al_{t_{k+1}} \bar{\mu}_{B^2}^{\delta} + \bar{\mu}_B^{\delta} \bigr) \bigl( {}^{(1)}\!A_{t_k}^{\delta} \bigr)^2 + 2\al_{t_{k+1}} \bar{\mu}_B^{\delta} {}^{(1)}\!A_{t_k}^{\delta} \Bigr] \\
&\qquad + 2\al_{t_{k+1}} \pi_{t_{k+1}}(1,1) \bar{\mu}_B^+ \bar{\mu}_B^- \, {}^{(1)}\!A_{t_k}^+ {}^{(1)}\!A_{t_k}^-,
\end{aligned}
\label{eq:placement_alpha_recursion_statement}
\end{equation}

\begin{equation}
\begin{aligned}
h_{t_k} &= h_{t_{k+1}} + \sum_{\delta=\pm} \pi_{t_{k+1}}^{\delta} \Bigg\{
2\bigl( \al_{t_{k+1}} \bar{\mu}_{B^2}^{\delta} + \bar{\mu}_B^{\delta} \bigr) \Bigl[ {}^{(1)}\!A_{t_k}^{\delta} \bigl( \delta {}^{(3)}\!A_{t_k}^{\delta} + {}^{(2)}\!A_{t_k}^{\delta} \bigr) \Bigr] \\
&\qquad + 2\al_{t_{k+1}} \bar{\mu}_B^{\delta} \bigl( \delta {}^{(3)}\!A_{t_k}^{\delta} + {}^{(2)}\!A_{t_k}^{\delta} \bigr) - 2\al_{t_{k+1}} \delta \bar{\mu}_{AB}^{\delta} \\
&\qquad + \delta {}^{(1)}\!A_{t_k}^{\delta} \Bigl( \bar{\mu}_A^{\delta} - \delta h_{t_{k+1}} \bar{\mu}_B^{\delta} + 2\al_{t_{k+1}} \bar{\mu}_{AB}^{\delta} \Bigr)
\Bigg\} \\
&\qquad - 2\al_{t_{k+1}} \pi_{t_{k+1}}(1,1) \bar{\mu}_B^+ \bar{\mu}_B^- \Bigg[
{}^{(1)}\!A_{t_k}^+ \bigl( {}^{(3)}\!A_{t_k}^- - {}^{(2)}\!A_{t_k}^- \bigr) - {}^{(1)}\!A_{t_k}^- \bigl( {}^{(2)}\!A_{t_k}^+ + {}^{(3)}\!A_{t_k}^+ \bigr) \\
&\qquad\qquad - \frac{\bar{\mu}_A^+}{\bar{\mu}_B^+} {}^{(1)}\!A_{t_k}^- + \frac{\bar{\mu}_A^-}{\bar{\mu}_B^-} {}^{(1)}\!A_{t_k}^+
\Bigg],
\end{aligned}
\label{eq:placement_h_recursion_statement}
\end{equation}

\begin{equation}
\begin{aligned}
g_{t_k} &= g_{t_{k+1}} + \sum_{\delta=\pm} \pi_{t_{k+1}}^{\delta} \Bigg\{
\bigl( \al_{t_{k+1}} \bar{\mu}_{B^2}^{\delta} + \bar{\mu}_B^{\delta} \bigr) \bigl( {}^{(3)}\!A_{t_k}^{\delta} + \delta {}^{(2)}\!A_{t_k}^{\delta} \bigr)^2 + \al_{t_{k+1}} \bar{\mu}_{A^2}^{\delta} \\
&\qquad - \delta h_{t_{k+1}} \bar{\mu}_A^{\delta}
+ \Bigl( \bar{\mu}_A^{\delta} - \delta h_{t_{k+1}} \bar{\mu}_B^{\delta} + 2\al_{t_{k+1}} \bar{\mu}_{AB}^{\delta} \Bigr)
\bigl( {}^{(3)}\!A_{t_k}^{\delta} + \delta {}^{(2)}\!A_{t_k}^{\delta} \bigr)
\Bigg\} \\
&\qquad - 2\al_{t_{k+1}} \pi_{t_{k+1}}(1,1) \bar{\mu}_B^+ \bar{\mu}_B^- \Bigg[
\bigl( {}^{(2)}\!A_{t_k}^+ + {}^{(3)}\!A_{t_k}^+ \bigr) \bigl( {}^{(3)}\!A_{t_k}^- - {}^{(2)}\!A_{t_k}^- \bigr) \\
&\qquad\qquad + \frac{\bar{\mu}_A^+}{\bar{\mu}_B^+} \bigl( {}^{(3)}\!A_{t_k}^- - {}^{(2)}\!A_{t_k}^- \bigr)
+ \frac{\bar{\mu}_A^-}{\bar{\mu}_B^-} \bigl( {}^{(2)}\!A_{t_k}^+ + {}^{(3)}\!A_{t_k}^+ \bigr)
+ \frac{\bar{\mu}_A^+ \bar{\mu}_A^-}{\bar{\mu}_B^+ \bar{\mu}_B^-}
\Bigg].
\end{aligned}
\label{eq:placement_g_recursion_statement}
\end{equation}

\begin{proof}  \\
We partition the number of market orders arriving during \([t_k,t_{k+1})\) into bins (blocks). For each bin \(j = 1,\dots,5\) (excluding the zero‑arrival case) we have random variables \(A_{t_{k+1}}^\de(j), B_{t_{k+1}}^\de(j)\) describing the demand when the order count falls into that bin. The probability of being in bin \(j\) is
\begin{equation}
p_j := \pi_{t_{k+1}}(1) \cdot P\bigl[N \in \text{Block}(j)\bigr],
\label{eq:placement_bin_weight}
\end{equation}
where \(\pi_{t_{k+1}}(1)\) is the probability of at least one arrival (of either side) and the conditional probability refers to the bin given that at least one arrival occurs. For \(N=0\) we have \(p_0=0\) and it is excluded.

The value function now takes the form
\begin{equation}
\begin{aligned}
V_{t_k} = \sup_{L_{t_k}^{\pm}} \sum_{j =1}^{5} \mathbb{E}\Bigg[ 
&\sum_{\delta = \pm} (S_{t_k} + \delta L_{t_k}^{\delta}) \, (A_{t_{k+1}}^{\delta}(j) + B_{t_{k+1}}^{\delta}(j)*L_{t_k}^{\delta}) \\
&+ \alpha_{t_{k+1}} \left( I_{t_k} - \sum_{\delta = \pm} (A_{t_{k+1}}^{\delta}(j) + B_{t_{k+1}}^{\delta}(j)*L_{t_k}^{\delta}) \right)^{\!2} \\
&+ S_{t_{k+1}} \left( I_{t_k} - \sum_{\delta = \pm} (A_{t_{k+1}}^{\delta}(j) + B_{t_{k+1}}^{\delta}(j)*L_{t_k}^{\delta}) \right) \\
&+ \widetilde{h}_{t_{k+1}} \left( I_{t_k} - \sum_{\delta = \pm} (A_{t_{k+1}}^{\delta}(j) + B_{t_{k+1}}^{\delta}(j)*L_{t_k}^{\delta}) \right) \\
&+ \widetilde{g}_{t_{k+1}} \;\Bigg|\; \mathcal{F}_{t_k}, \mathbbm{1}_\tkk^{\delta}=1 \Bigg]* \pi_{t_{k+1}}(1) * P[N \in Block(j)]
\end{aligned}
\label{eq:placement_bin_value_expansion}
\end{equation}
We keep the same martingale assumption for \(S_{t_k}\) and the same ansatz for the value function:
\begin{equation}
v(t_k,S_{t_k},W_{t_k},I_{t_k}) = W_{t_k} + \alpha_{t_k} I_{t_k}^2 + S_{t_k}I_{t_k} + h_{t_k}I_{t_k} + g_{t_k}
\label{eq:placement_bin_candidate}
\end{equation}
with terminal conditions \(\alpha_T=-\lambda,\;h_T=0,\;g_T=0\).

\textbf{Step 1: Bellman equation and first‑order conditions} \\
Substituting the ansatz into the Bellman equation \(v(t_k,\cdot)=\sup_{L^\pm}\E[v(t_{k+1},\cdot)|\F_{t_k}]\) and using the martingale property as before, we obtain for a fixed bin \(j\) an expression identical to the one derived in the original proof, but with \(A^\de,B^\de\) replaced by their bin‑specific versions. Summing over bins with weights \(p_j\) gives the total objective to be maximized.

Define for each bin \(j\) the conditional moments (given \(\F_{t_k}\) and that the bin is \(j\)):
\begin{equation}
\begin{aligned}
\mu_{A}^\de(j) &:= \E\bigl[A_{t_{k+1}}^\de(j) \big| \F_{t_k}, \text{bin } j\bigr],
\mu_{B}^\de(j) := \E\bigl[B_{t_{k+1}}^\de(j) \big| \F_{t_k}, \text{bin } j\bigr],\\
\mu_{A^2}^\de(j) &:= \E\bigl[(A_{t_{k+1}}^\de(j))^2 \big| \F_{t_k}, \text{bin } j\bigr],
\mu_{B^2}^\de(j):= \E\bigl[(B_{t_{k+1}}^\de(j))^2 \big| \F_{t_k}, \text{bin } j\bigr],\\
\mu_{AB}^\de(j) &:= \E\bigl[A_{t_{k+1}}^\de(j) B_{t_{k+1}}^\de(j) \big| \F_{t_k}, \text{bin } j\bigr].
\end{aligned}
\label{eq:placement_bin_conditional_moments}
\end{equation}
Within each bin the arrival indicators are deterministic (the bin specifies the number of orders, so if the bin corresponds to a positive number, then the corresponding side definitely arrives). Therefore, for a given bin \(j\) we have \(\1_{t_{k+1}}^\de = 1\) for the sides that are active. To keep the notation manageable, we assume that each bin \(j\) fully determines which sides are active; the probabilities \(\pi^\de\) and \(\pi(1,1)\) are then replaced by indicators. However, a more systematic approach is to work directly with the weighted sums.

Define the total probability of at least one arrival on side \(\delta\) as
\begin{equation}
\pi^{\delta} := \sum_{j} p_j^{\delta},
\label{eq:placement_marginal_arrival_probability}
\end{equation}
where \(p_j^{\delta}\) is the probability that bin \(j\) contains at least one order on side \(\delta\). Similarly,
\begin{equation}
\pi(1,1) := \sum_{j} p_j^{\text{both}}
\label{eq:placement_joint_arrival_probability}
\end{equation}
is the probability of simultaneous buy and sell arrivals.

The mixed moments are defined as:
\begin{equation}
\begin{aligned}
\bar{\mu}_B^{\delta} &:= \frac{1}{\pi^{\delta}}\sum_{j} p_j^{\delta} \mu_B^{\delta}(j), \quad
\bar{\mu}_{B^2}^{\delta} := \frac{1}{\pi^{\delta}}\sum_{j} p_j^{\delta} \mu_{B^2}^{\delta}(j),\\
\bar{\mu}_A^{\delta} &:= \frac{1}{\pi^{\delta}}\sum_{j} p_j^{\delta} \mu_A^{\delta}(j), \quad
\bar{\mu}_{AB}^{\delta} := \frac{1}{\pi^{\delta}}\sum_{j} p_j^{\delta} \mu_{AB}^{\delta}(j),\\
\bar{\mu}_{A^2}^{\delta} &:= \frac{1}{\pi^{\delta}}\sum_{j} p_j^{\delta} \mu_{A^2}^{\delta}(j).
\end{aligned}
\label{eq:placement_mixed_moments}
\end{equation}
where \(\pi^\de = \sum_j p_j^\de\) is the total probability of at least one arrival on side \(\de\) (note that \(\pi^\de\) may differ from \(\pi(1)\) if only one side can arrive in some bins). In practice, the bins should be defined jointly for both sides so that the probabilities \(\pi^\de\) and \(\pi(1,1)\) are naturally obtained from the bin probabilities.

We adopt the same ansatz for the value function as in the original paper:
\begin{equation}
v(t_k,S_{t_k},W_{t_k},I_{t_k}) = W_{t_k} + \alpha_{t_k} I_{t_k}^2 + S_{t_k} I_{t_k} + h_{t_k} I_{t_k} + g_{t_k},
\label{eq:placement_mixed_candidate}
\end{equation}
with terminal conditions \(\alpha_T = -\lambda\), \(h_T = 0\), \(g_T = 0\).
Using the dynamics
\begin{equation}
W_{t_{k+1}} = W_{t_k} + \sum_{\delta=\pm} (S_{t_k} + \delta L_{t_k}^{\delta}) \, Q_{t_{k+1}}^{\delta},\qquad
I_{t_{k+1}} = I_{t_k} - \sum_{\delta=\pm} \delta Q_{t_{k+1}}^{\delta},
\label{eq:placement_bin_dynamics}
\end{equation}
where \(Q_{t_{k+1}}^{\delta} = A_{t_{k+1}}^{\delta}(j) + B_{t_{k+1}}^{\delta}(j) L_{t_k}^{\delta}\) when the bin is \(j\), the Bellman equation becomes
\begin{equation}
v(t_k,\cdot) = \sup_{L_{t_k}^{\pm}} \sum_{j} p_j \, \E\Bigl[ v(t_{k+1}, S_{t_{k+1}}, W_{t_{k+1}}, I_{t_{k+1}}) \,\big|\, \F_{t_k}, \text{bin } j \Bigr].
\label{eq:placement_bin_bellman}
\end{equation}
Substituting the ansatz and using the martingale property \(\E[S_{t_{k+1}}|\F_{t_k}] = S_{t_k}\) (which can be relaxed later), we obtain after straightforward algebra:
\begin{equation}
f(L^+,L^-) := \E\bigl[ v(t_{k+1},\cdot) \,\big|\, \F_{t_k} \bigr] - \bigl(W_{t_k} + \alpha_{t_k} I_{t_k}^2 + S_{t_k} I_{t_k} + h_{t_k} I_{t_k} + g_{t_k}\bigr)
\label{eq:placement_objective_difference}
\end{equation}
\begin{equation}
\begin{aligned}
= &\sum_{\delta=\pm} \pi_{t_{k+1}}^{\delta} \Bigl[ \bigl(\alpha_{t_{k+1}} \bar{\mu}_{B^2}^{\delta} + \bar{\mu}_B^{\delta}\bigr) (L^{\delta})^2
+ \bigl( \bar{\mu}_A^{\delta} - \delta h_{t_{k+1}} \bar{\mu}_B^{\delta} + 2\alpha_{t_{k+1}}(\delta \bar{\mu}_B^{\delta} I_{t_k} - \bar{\mu}_{AB}^{\delta}) \bigr) L^{\delta} \Bigr] \\
& - 2\alpha_{t_{k+1}} \pi_{t_{k+1}}(1,1) \bar{\mu}_B^+ \bar{\mu}_B^- L^+ L^- \\
& + \text{terms independent of } L^+,L^-.
\end{aligned}
\label{eq:placement_quadratic_expansion}
\end{equation}
This is a quadratic function in \(L^+\) and \(L^-\). Therefore the optimal controls and the backward recursions are identical in structure.

Thus the problem reduces exactly to the one solved in the original proof, with the understanding that all moments are now the mixed moments defined above. Consequently, the first‑order conditions, the expressions for the optimal controls, and the recursions for \(\al_{t_k}, h_{t_k}, g_{t_k}\) are \("identical in form\)" to those given in the original paper. We reproduce them here for completeness.

\textbf{Step 2: Optimal controls and recursions} \\
we define
\begin{equation}
\ga_{t_k} := \bigl[ \pi_{t_{k+1}}(1,1) \al_{t_{k+1}} \bar{\mu}_B^+ \bar{\mu}_B^- \bigr]^2 - \pi_{t_{k+1}}^+ \pi_{t_{k+1}}^- \bigl( \al_{t_{k+1}} \bar{\mu}_{B^2}^+ + \bar{\mu}_B^+ \bigr) \bigl( \al_{t_{k+1}} \bar{\mu}_{B^2}^- + \bar{\mu}_B^- \bigr)
\label{eq:placement_gamma}
\end{equation} 
\begin{equation}
\be_{t_k}^{\pm} := \pi_{t_{k+1}}^+ \pi_{t_{k+1}}^- \bar{\mu}_B^{\pm} \bigl( \al_{t_{k+1}} \bar{\mu}_{B^2}^{\mp} + \bar{\mu}_B^{\mp} \bigr) - \pi_{t_{k+1}}^{\mp} \pi_{t_{k+1}}(1,1) \al_{t_{k+1}} \bar{\mu}_B^{\pm} (\bar{\mu}_B^{\mp})^2.
\label{eq:placement_beta}
\end{equation} 
The first-order conditions yield:
\begin{equation}
\begin{aligned}
L_{t_k}^{+,*} &= {}^{(1)}\!A_{t_k}^+ I_{t_k} + {}^{(2)}\!A_{t_k}^+ + {}^{(3)}\!A_{t_k}^+,\\
L_{t_k}^{-,*} &= -{}^{(1)}\!A_{t_k}^- I_{t_k} - {}^{(2)}\!A_{t_k}^- + {}^{(3)}\!A_{t_k}^-,
\end{aligned}
\label{eq:placement_optimal_quotes}
\end{equation}
where
\begin{equation}
\begin{aligned}
{}^{(1)}\!A_{t_k}^{\pm} &= \frac{\be_{t_k}^{\pm} \al_{t_{k+1}}}{\ga_{t_k}},\\
{}^{(2)}\!A_{t_k}^{\pm} &= \frac{\be_{t_k}^{\pm} h_{t_{k+1}}}{2\ga_{t_k}},\\
{}^{(3)}\!A_{t_k}^{\pm} &= \frac{\pi_{t_{k+1}}^{\mp}}{2\ga_{t_k}} \bigl( \al_{t_{k+1}} \bar{\mu}_{B^2}^{\mp} + \bar{\mu}_B^{\mp} \bigr) \Bigl[ \pi_{t_{k+1}}^{\pm} \bigl( \bar{\mu}_A^{\pm} + 2\al_{t_{k+1}} \bar{\mu}_{AB}^{\pm} \bigr) - 2\al_{t_{k+1}} \pi_{t_{k+1}}(1,1) \bar{\mu}_B^{\pm} \bar{\mu}_A^{\mp} \Bigr] \\
&\quad + \pi_{t_{k+1}}(1,1) \frac{\al_{t_{k+1}}}{2\ga_{t_k}} \bar{\mu}_B^+ \bar{\mu}_B^- \Bigl[ \pi_{t_{k+1}}^{\mp} \bigl( \bar{\mu}_A^{\mp} + 2\al_{t_{k+1}} \bar{\mu}_{AB}^{\mp} \bigr) - 2\al_{t_{k+1}} \pi_{t_{k+1}}(1,1) \bar{\mu}_B^{\mp} \bar{\mu}_A^{\pm} \Bigr].
\end{aligned}
\label{eq:placement_quote_coefficients}
\end{equation}
Matching powers of \(I_{t_k}\) in the Bellman equation gives the recursions for \(\alpha_{t_k}\), \(h_{t_k}\), and \(g_{t_k}\):
\begin{equation}
\boxed{
\begin{aligned}
\al_{t_k} &= \al_{t_{k+1}} + \sum_{\delta=\pm} \pi_{t_{k+1}}^{\delta} \Bigl[ \bigl( \al_{t_{k+1}} \bar{\mu}_{B^2}^{\delta} + \bar{\mu}_B^{\delta} \bigr) \bigl( {}^{(1)}\!A_{t_k}^{\delta} \bigr)^2 + 2\al_{t_{k+1}} \bar{\mu}_B^{\delta} {}^{(1)}\!A_{t_k}^{\delta} \Bigr] \\
&\qquad + 2\al_{t_{k+1}} \pi_{t_{k+1}}(1,1) \bar{\mu}_B^+ \bar{\mu}_B^- \, {}^{(1)}\!A_{t_k}^+ {}^{(1)}\!A_{t_k}^-,
\end{aligned}
}
\label{eq:placement_alpha_recursion}
\end{equation}

\begin{equation}
\boxed{
\begin{aligned}
h_{t_k} &= h_{t_{k+1}} + \sum_{\delta=\pm} \pi_{t_{k+1}}^{\delta} \Bigg\{
2\bigl( \al_{t_{k+1}} \bar{\mu}_{B^2}^{\delta} + \bar{\mu}_B^{\delta} \bigr) \Bigl[ {}^{(1)}\!A_{t_k}^{\delta} \bigl( \delta {}^{(3)}\!A_{t_k}^{\delta} + {}^{(2)}\!A_{t_k}^{\delta} \bigr) \Bigr] \\
&\qquad + 2\al_{t_{k+1}} \bar{\mu}_B^{\delta} \bigl( \delta {}^{(3)}\!A_{t_k}^{\delta} + {}^{(2)}\!A_{t_k}^{\delta} \bigr) - 2\al_{t_{k+1}} \delta \bar{\mu}_{AB}^{\delta} \\
&\qquad + \delta {}^{(1)}\!A_{t_k}^{\delta} \Bigl( \bar{\mu}_A^{\delta} - \delta h_{t_{k+1}} \bar{\mu}_B^{\delta} + 2\al_{t_{k+1}} \bar{\mu}_{AB}^{\delta} \Bigr)
\Bigg\} \\
&\qquad - 2\al_{t_{k+1}} \pi_{t_{k+1}}(1,1) \bar{\mu}_B^+ \bar{\mu}_B^- \Bigg[
{}^{(1)}\!A_{t_k}^+ \bigl( {}^{(3)}\!A_{t_k}^- - {}^{(2)}\!A_{t_k}^- \bigr) - {}^{(1)}\!A_{t_k}^- \bigl( {}^{(2)}\!A_{t_k}^+ + {}^{(3)}\!A_{t_k}^+ \bigr) \\
&\qquad\qquad - \frac{\bar{\mu}_A^+}{\bar{\mu}_B^+} {}^{(1)}\!A_{t_k}^- + \frac{\bar{\mu}_A^-}{\bar{\mu}_B^-} {}^{(1)}\!A_{t_k}^+
\Bigg],
\end{aligned}
}
\label{eq:placement_h_recursion}
\end{equation}

\begin{equation}
\boxed{
\begin{aligned}
g_{t_k} &= g_{t_{k+1}} + \sum_{\delta=\pm} \pi_{t_{k+1}}^{\delta} \Bigg\{
\bigl( \al_{t_{k+1}} \bar{\mu}_{B^2}^{\delta} + \bar{\mu}_B^{\delta} \bigr) \bigl( {}^{(3)}\!A_{t_k}^{\delta} + \delta {}^{(2)}\!A_{t_k}^{\delta} \bigr)^2 + \al_{t_{k+1}} \bar{\mu}_{A^2}^{\delta} \\
&\qquad - \delta h_{t_{k+1}} \bar{\mu}_A^{\delta}
+ \Bigl( \bar{\mu}_A^{\delta} - \delta h_{t_{k+1}} \bar{\mu}_B^{\delta} + 2\al_{t_{k+1}} \bar{\mu}_{AB}^{\delta} \Bigr)
\bigl( {}^{(3)}\!A_{t_k}^{\delta} + \delta {}^{(2)}\!A_{t_k}^{\delta} \bigr)
\Bigg\} \\
&\qquad - 2\al_{t_{k+1}} \pi_{t_{k+1}}(1,1) \bar{\mu}_B^+ \bar{\mu}_B^- \Bigg[
\bigl( {}^{(2)}\!A_{t_k}^+ + {}^{(3)}\!A_{t_k}^+ \bigr) \bigl( {}^{(3)}\!A_{t_k}^- - {}^{(2)}\!A_{t_k}^- \bigr) \\
&\qquad\qquad + \frac{\bar{\mu}_A^+}{\bar{\mu}_B^+} \bigl( {}^{(3)}\!A_{t_k}^- - {}^{(2)}\!A_{t_k}^- \bigr)
+ \frac{\bar{\mu}_A^-}{\bar{\mu}_B^-} \bigl( {}^{(2)}\!A_{t_k}^+ + {}^{(3)}\!A_{t_k}^+ \bigr)
+ \frac{\bar{\mu}_A^+ \bar{\mu}_A^-}{\bar{\mu}_B^+ \bar{\mu}_B^-}
\Bigg].
\end{aligned}
}
\label{eq:placement_g_recursion}
\end{equation}
These recursions, together with the terminal conditions, fully characterize the optimal market making strategy under the binned arrival model.
\end{proof}


\section{Parameters}

Estimating the base intensity vector $\nu$ and the excitation kernel $k$ of an
ordinary Hawkes process has attracted considerable attention in the statistics
and financial econometrics literature. While recent developments have focused
on non-parametric approaches~\cite{bacry2016first, kirchner2017estimation,
eichler2017graphical, sancetta2018estimation, achab2018analysis}, the present
paper adopts the classical parametric formulation  and extends it to
state-dependent Hawkes processes. In the sequel we consider a parametric kernel
$k = k_{\theta}$ with $\theta \in \mathbb{R}^p$, and defer the non-parametric
setting to future work.

We use the classical Maximum Likelihood to estimate the parameters. The log-likelihood function for a state-dependent Hawkes process with parameters
$(\phi, \nu, \theta)$ factorises as
\begin{equation}
\ln L(\phi, \nu, \theta)
    = \sum_{n=1}^{N} \ln \phi_{e_n}(x_{n-1}, x_n)
      + \left( \sum_{n=1}^{N} \ln \lambda_{e_n}(t_n)
             - \int_{0}^{T} \sum_{e \in \mathcal{E}} \lambda_e(t) \, dt \right),
\label{eq:loglik}
\end{equation}

where the intensity of event type $e$ is
\begin{equation}
\lambda_e(t) = \nu_e + \sum_{e' \in \mathcal{E}} \int_{(0,t)} k_{e'e}(t-s,\, X(s)) \, dN_{e'}(s),
\label{eq:intensity}
\end{equation}
and the excitation kernel $k_{e'e}(t,x)$ is parametrised in the exponential form
\begin{equation}
k_{e'e}(t,x) = \alpha_{e' x e} \, e^{-\beta_{e' x e}\, t},
\qquad \alpha_{e' x e},\, \beta_{e' x e} \ge 0.
\label{eq:exponential_kernel}
\end{equation}

The factorisation in~\eqref{eq:loglik} decouples the transition parameters
$\phi$ from the intensity parameters $(\nu, \theta)$, so the two blocks are
estimated in sequence. For the transition probabilities, maximising the first
term of~\eqref{eq:loglik} under the constraint that each $\phi_e$ is a
stochastic matrix yields the empirical transition-count estimator
\begin{equation}
\hat{\phi}_e(x, x')
    = \frac{ \#\{ n : e_n = e,\; x_{n-1} = x,\; x_n = x' \} }
           { \#\{ n : e_n = e,\; x_{n-1} = x \} }.
\label{eq:phi_hat}
\end{equation}
The intensity parameters $(\nu, \theta)$ are then obtained by maximising the
second term of~\eqref{eq:loglik}; for the exponential kernel of
\eqref{eq:exponential_kernel} the integral term admits a closed-form expression
that reduces the optimisation to a finite-dimensional problem in the event
times and states, and the log-intensity sum can be evaluated by the standard
recursion of Appendix~A.2 of~\cite{morariu2022state}:

\begin{equation}
(\hat{\nu}, \hat{\theta}) = \arg\max_{\nu,\theta} \left( \sum_{n=1}^{N} \ln \lambda_{e_n}(t_n) - \int_0^T \sum_{e \in \mathcal{E}} \lambda_e(t) \, dt \right).
\end{equation}
For exponential kernels \eqref{eq:exponential_kernel}, the integral term can be explicitly computed as

\begin{equation}
\int_0^T \sum_{e \in \mathcal{E}} \lambda_e(t) \, dt = \sum_{e \in \mathcal{E}} \nu_e T + \sum_{e' \in \mathcal{E}} \sum_{x \in \mathcal{X}} \sum_{e \in \mathcal{E}} \frac{\alpha_{e' x e}}{\beta_{e' x e}} \, N_{e' x}(T),
\end{equation}
where \(N_{e' x}(T)\) is the total number of events of type \(e'\) that occurred in state \(x\) up to time \(T\).

\section{Numerical Validation}


\subsection{Data Profile}
\label{sec:data}

We calibrate and test the model using high-frequency LOB
data of MSFT, recorded across five consecutive
trading days from November~3 to November~7, 2014.
Following standard practice in the high-frequency trading literature
\citep{chavez2024adaptive,capponi2021market},
we restrict attention to regular trading hours from 9:30\,a.m.\ to
4:00\,p.m.\ Eastern Time, yielding $23{,}400$ seconds of trading per day
and $117{,}000$ one-second intervals in total across the five-day sample.

\paragraph{Order classification.}
Each event is classified into one of four order types that drive the
state-dependent Hawkes process:
\begin{enumerate}
  \item \emph{Market Buy (MB):} a marketable order that consumes liquidity on
        the ask side (message types 69 and 70 with \texttt{BidSideFlag}~$= 0$);
  \item \emph{Market Sell (MS):} a marketable order that consumes liquidity on
        the bid side (message types 69 and 70 with \texttt{BidSideFlag}~$= 1$);
  \item \emph{Limit Buy (LB):} a passive order posted to the bid side
        (message type~66);
  \item \emph{Limit Sell (LS):} a passive order posted to the ask side
        (message type~83).
\end{enumerate}
Table~\ref{tab:data_summary} reports aggregate statistics for the
five-day sample.

\begin{table}[ht]
\centering
\caption{Summary statistics for MSFT LOB data,
         November~3--7, 2014 (five trading days).}
\label{tab:data_summary}
\begin{tabular}{lrr}
\toprule
Statistic & Buy Side & Sell Side \\
\midrule
Total events (all types)         & \multicolumn{2}{c}{5{,}351{,}271} \\
Market orders (MB $+$ MS)        & ${\approx}74{,}177$ & ${\approx}56{,}110$ \\
\quad Total market orders         & \multicolumn{2}{c}{130{,}287} \\
Limit orders (LB $+$ LS)         & \multicolumn{2}{c}{2{,}651{,}483} \\
\midrule
Avg.\ market order size (shares) & 203.63 & 206.87 \\
Overall avg.\ MO size (shares)   & \multicolumn{2}{c}{205.02} \\
\midrule
\multicolumn{3}{l}{\textit{Per-interval statistics ($\Delta t = 1$ s)}} \\
Total 1-s intervals              & \multicolumn{2}{c}{117{,}000} \\
Intervals with zero MOs          & 105{,}795\ (90.42\%) & 106{,}245\ (90.81\%) \\
Intervals with $\geq 1$ MO       & 11{,}205\ (9.58\%)  & 10{,}755\ (9.19\%) \\
Avg.\ MOs $|$ interval $\geq 1$  & 6.62 & 5.22 \\
Overall avg.\ MOs per interval   & 0.6340 & 0.4796 \\
\bottomrule
\end{tabular}
\end{table}

\paragraph{Market order arrival profile.}
A salient feature of the data is the high degree of
\emph{zero-inflation}: across all physical time steps,
the majority of intervals contain no market orders on either side.
Figure~\ref{fig:mo_dist_1s} displays the empirical frequency distribution
of per-interval market order counts $N$ for $\De t = 1$\,s, restricted
to intervals with $N \geq 1$.
The distribution is heavily right-skewed, with the bulk of active intervals
containing between one and ten orders.
Figure~\ref{fig:zero_inflation} shows how the zero-fraction decreases as
$\De t$ grows: from approximately 90\% at $\De t = 1$\,s to roughly 77\%
at $\De t = 3$\,s, reflecting the increasing probability of observing at
least one market order in a longer window.

\begin{figure}[ht]
  \centering
  \includegraphics[width=0.82\textwidth]{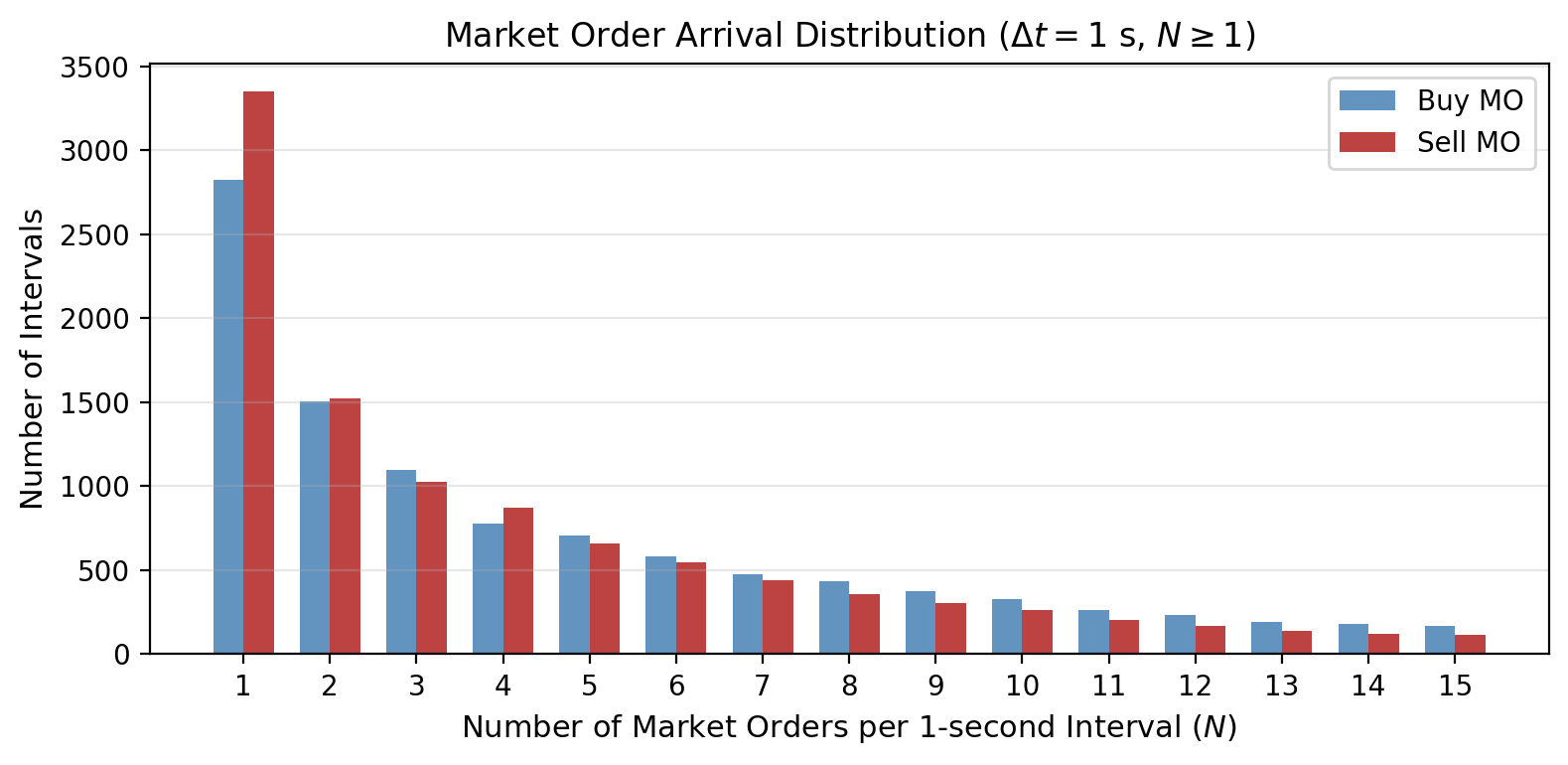}
  \caption{Empirical frequency distribution of per-interval market order
    arrivals for $\Delta t = 1$\,s, restricted to intervals containing at
    least one market order ($N \geq 1$, shown up to $N = 15$).
    Buy and sell sides are displayed separately.
    Intervals with $N = 0$ account for 90.4\% (buy) and 90.8\% (sell)
    of all $117{,}000$ one-second intervals.}
  \label{fig:mo_dist_1s}
\end{figure}

\begin{figure}[ht]
  \centering
  \includegraphics[width=0.75\textwidth]{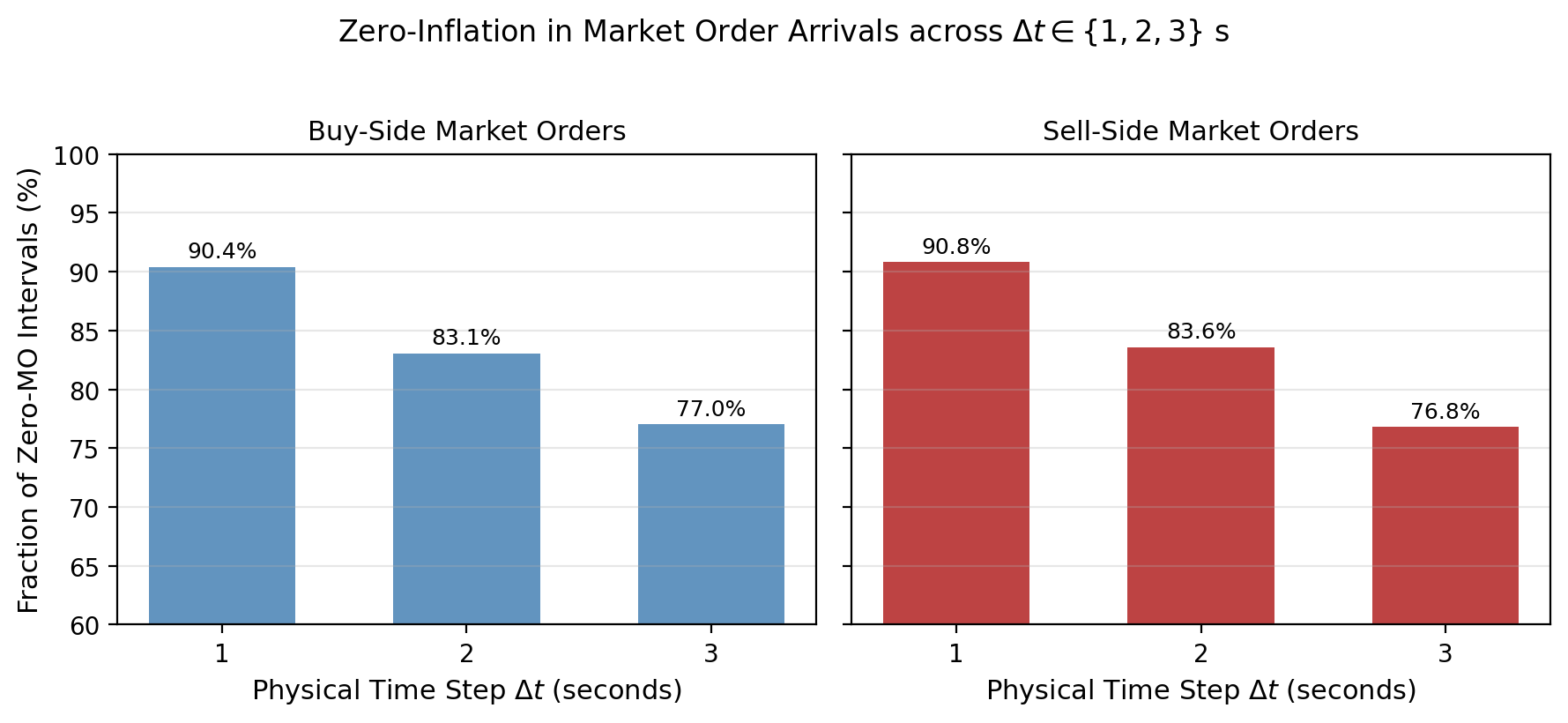}
  \caption{Fraction of intervals containing zero market orders
    as a function of the physical time step $\Delta t \in \{1,2,3\}$\,s,
    for the buy side (left) and sell side (right).
    The zero fraction decreases monotonically with $\Delta t$,
    from approximately 90\% at $\Delta t = 1$\,s to 77\% at $\Delta t = 3$\,s.}
  \label{fig:zero_inflation}
\end{figure}

\paragraph{Block grouping.}
Figure~\ref{fig:block_probs} reports the empirical probabilities of each
block. The dominant block is $N = 0$, which accounts for over 90\% of all intervals.
The remaining probability mass is concentrated in the $[1,10]$ block
(approximately 7.8\% for buy and 8.0\% for sell),
with the higher-activity blocks each contributing less than 2\%.

\begin{figure}[ht]
  \centering
  \includegraphics[width=0.88\textwidth]{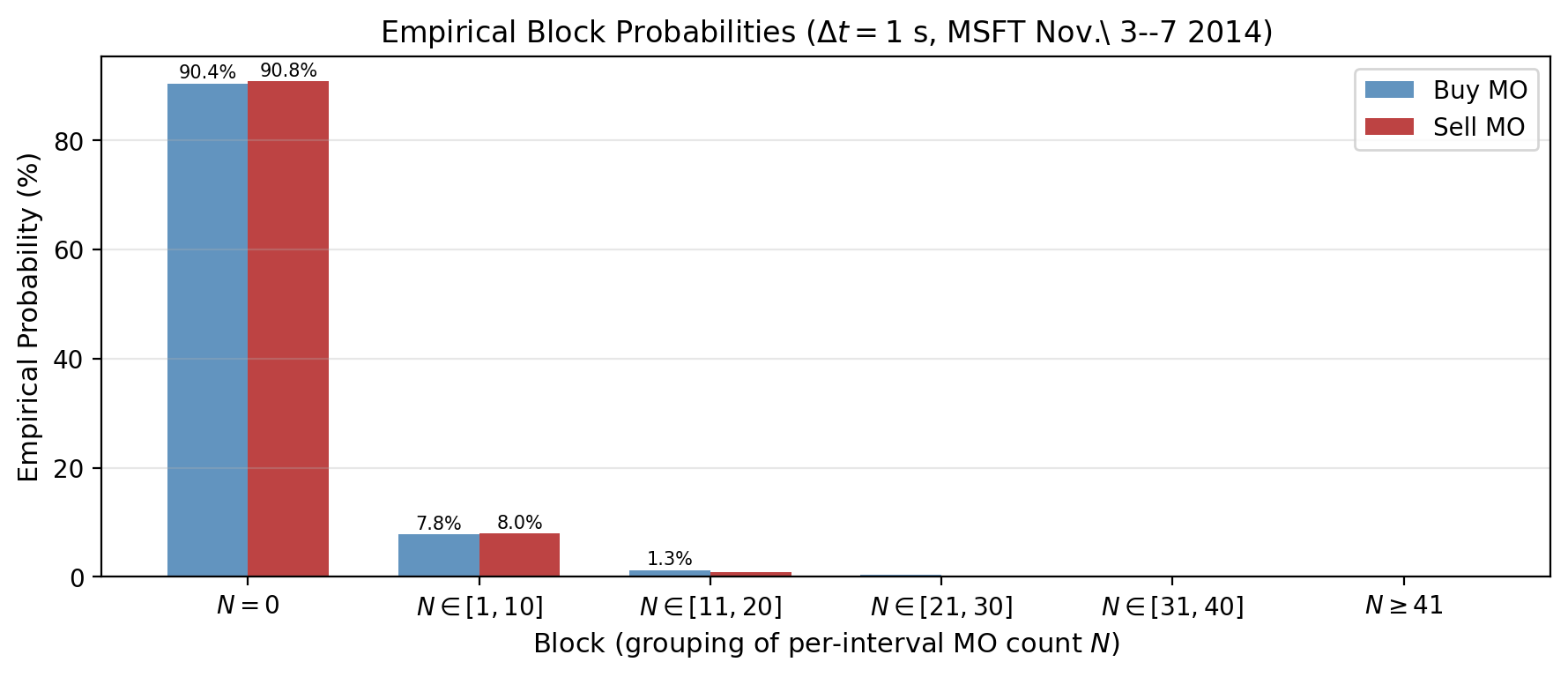}
  \caption{Empirical block probabilities for the grouping of per-interval
    market order counts ($\De t = 1$\,s, MSFT, November~3--7, 2014).
    Blocks are $\{0\}$, $[1,10]$, $[11,20]$, $[21,30]$, $[31,40]$, and
    $[41,\infty)$; buy and sell sides are displayed side by side.}
  \label{fig:block_probs}
\end{figure}

\paragraph{State variables.}
The state process $X(\tk)$ entering the intensity function of the
state-dependent Hawkes model is constructed from three LOB features
observed at each physical time step $\tk$:
\begin{enumerate}
  \item \emph{Queue imbalance:} the ratio
        $\displaystyle\frac{V^{\mathrm{bid}} - V^{\mathrm{ask}}}
                            {V^{\mathrm{bid}} + V^{\mathrm{ask}}}$,
        where $V^{\mathrm{bid}}$ and $V^{\mathrm{ask}}$ denote the
        total volume at the best bid and best ask, respectively;
\end{enumerate}
This feature is computed from the snapshot of the LOB at the
end of each physical time interval and serves as the conditioning
state for the sdHawkes intensity in the subsequent sections.

\subsection{Demand Function Results}
\label{sec:demand_results}

We now present the empirical estimation of the parametric market order demand function. The coefficients $a_N^\pm$ (intercept) and $b_N^\pm$ (slope) are estimated separately for each group of $N$, as detailed below.

\subsubsection{Market Order Frequency Distribution}

For each second interval, we count the number of buy and sell MO separately.  Table~\ref{tab:mo_freq} summarizes the frequency distribution of MO counts per second.  Approximately $90.4\%$ of one-second intervals contain no buy MO, and a similar proportion holds for sell MOs.  This high proportion of empty intervals is expected at the one-second resolution and motivates the use of a point process framework to model the sporadic and self-exciting nature of order arrivals.

\begin{table}[htbp]
\centering
\caption{Market order count frequency distribution per one-second interval (MSFT, Nov~3--7, 2014).  Total intervals: $117{,}000$.}
\label{tab:mo_freq}
\begin{tabular}{@{}rrcrrc@{}}
\toprule
\multicolumn{3}{c}{\textbf{Buy Orders}} & \multicolumn{3}{c}{\textbf{Sell Orders}} \\
\cmidrule(lr){1-3} \cmidrule(lr){4-6}
$N$ & Freq.\ & \%      & $N$ & Freq.\ & \% \\
\midrule
0      & 105,795 & 90.42\% & 0      & 105,184 & 89.90\% \\
1      & 2,826   & 2.42\%  & 1      & 2,874   & 2.46\%  \\
2      & 1,504   & 1.29\%  & 2      & 1,587   & 1.36\%  \\
3      & 1,096   & 0.94\%  & 3      & 1,113   & 0.95\%  \\
4      & 777     & 0.66\%  & 4      & 788     & 0.67\%  \\
5      & 708     & 0.61\%  & 5      & 731     & 0.62\%  \\
6--10  & 2,612   & 2.23\%  & 6--10  & 2,891   & 2.47\%  \\
11--20 & 1,487   & 1.27\%  & 11--20 & 1,803   & 1.54\%  \\
21--30 & 404     & 0.35\%  & 21--30 & 502     & 0.43\%  \\
31--40 & 121     & 0.10\%  & 31--40 & 127     & 0.11\%  \\
41+    & 90      & 0.08\%  & 41+    & 95      & 0.08\%  \\
\midrule
Total (nonzero) & 11,205 & 9.58\% & Total (nonzero) & 11,816 & 10.10\% \\
\bottomrule
\end{tabular}
\end{table}

\subsubsection{Execution Probability Estimation}

For each interval in which the MO count $N$ falls within a given block, we compute the execution probability of a hypothetical limit order posted at each level of the limit order book (LOB).  We adopt a \emph{front-of-queue priority} convention: the first MO arriving in the interval is matched against the available liquidity at each price level, starting from Level~1 (or the best bid/ask).  The execution probability at level $\ell$ is the probability that the total marketable share volume in that interval is sufficient to consume all liquidity at levels $1, \dots, \ell-1$ and reach level $\ell$, given our simulated order size of $\text{AVOS}=205$ shares.

Formally, for an interval with $N$ buy MOs arriving on the ask side, let $V_{\text{total}}$ be the total share volume of those $N$ orders, and let $D_\ell$ denote the depth (shares) at sell-side LOB level $\ell$.  The execution indicator at level $\ell$ is
\[
\mathbbm{1}_{\text{exec}}^\ell = \mathbbm{1}\!\left(V_{\text{total}} - \sum_{i=1}^{\ell-1} D_i > 0\right),
\]
and the fill proportion is $\min\!\bigl(1,\; (V_{\text{total}} - \sum_{i=1}^{\ell-1} D_i) / \text{AVOS}\bigr)$.  The execution probability at level $\ell$ for a given block $j$ is the average of these fill proportions over all intervals belonging to that block.  We consider the first five LOB levels ($\ell = 1, \dots, 5$).

Intervals with $N = 0$ are excluded from this analysis, since the demand function $f_N^\pm(L)$ is defined only conditional on the arrival of at least one MO.  The probability $P[N=0]$ is accounted for separately in the Bellman recursion of Section~4 via the arrival indicator $\mathbbm{1}_{t_{k+1}}^\pm$.

Table~\ref{tab:exec_prob} reports the estimated execution probability matrices for buy and sell orders, aggregated by block.  Figure~\ref{fig:exec_heatmap} displays the same information in heatmap form, and Figure~\ref{fig:exec_curves} shows the execution probability curves as functions of the LOB level.

\begin{table}[htbp]
\centering
\caption{Execution probability matrices by market order count group (five LOB levels, $\De t = 1$s).  Each entry reports the average probability that a limit order of size $\text{AVOS}=205$ shares posted at that level would be fully filled, conditional on the indicated block of $N$ MOs arriving in that interval.}
\label{tab:exec_prob}
\begin{tabular}{@{}lrrrrr@{}}
\toprule
\multicolumn{6}{c}{\textbf{Buy Orders}} \\
\cmidrule(lr){1-6}
Group ($N$) & Level 1 & Level 2 & Level 3 & Level 4 & Level 5 \\
\midrule
 1--10 & 0.787 & 0.062 & 0.005 & 0.002 & 0.001 \\
11--20 & 0.918 & 0.084 & 0.009 & 0.004 & 0.002 \\
21--30 & 0.939 & 0.121 & 0.017 & 0.006 & 0.004 \\
31--40 & 0.959 & 0.160 & 0.029 & 0.013 & 0.004 \\
  40+  & 0.942 & 0.169 & 0.043 & 0.017 & 0.008 \\
\bottomrule
\end{tabular}
\quad
\begin{tabular}{@{}lrrrrr@{}}
\toprule
\multicolumn{6}{c}{\textbf{Sell Orders}} \\
\cmidrule(lr){1-6}
Group ($N$) & Level 1 & Level 2 & Level 3 & Level 4 & Level 5 \\
\midrule
 1--10 & 0.817 & 0.067 & 0.006 & 0.002 & 0.001 \\
11--20 & 1.000 & 0.326 & 0.045 & 0.012 & 0.007 \\
21--30 & 1.000 & 0.550 & 0.116 & 0.043 & 0.013 \\
31--40 & 1.000 & 0.755 & 0.302 & 0.132 & 0.057 \\
  40+  & 1.000 & 0.893 & 0.429 & 0.286 & 0.071 \\
\bottomrule
\end{tabular}
\end{table}

Several observations are immediate from Table~\ref{tab:exec_prob}:
\begin{enumerate}
    \item \textbf{Monotonicity in level.}  For every block, the execution probability strictly decreases with the LOB level.  This confirms the basic economic intuition: limit orders posted farther from the mid-price (deeper in the book) are less likely to be executed.
    \item \textbf{Monotonicity in $N$.}  For a fixed level, the execution probability increases with the number of arriving MOs.  This is expected, since more market orders imply greater aggregate marketable volume, increasing the likelihood of penetrating deeper levels.
    \item \textbf{Saturation at Level~1.}  For blocks with $N \ge 11$, the Level~1 execution probability for sell-side orders reaches $1.00$, meaning that intervals with more than ten MOs almost always consume at least 205 shares at the best price on the bid side.  The buy side shows a similar but slightly lower saturation ($\approx 0.92$--$0.96$ for $N \ge 11$).
    \item \textbf{Symmetry.}  The buy-side and sell-side execution probabilities are broadly symmetric, with sell-side probabilities slightly higher at deeper levels, reflecting the small asymmetry in average order sizes (buy: 203.63 shares, sell: 206.87 shares; see Section~5.1).
\end{enumerate}

\begin{figure}[htbp]
    \centering
    \begin{subfigure}[b]{0.48\textwidth}
        \centering
        \includegraphics[width=\textwidth]{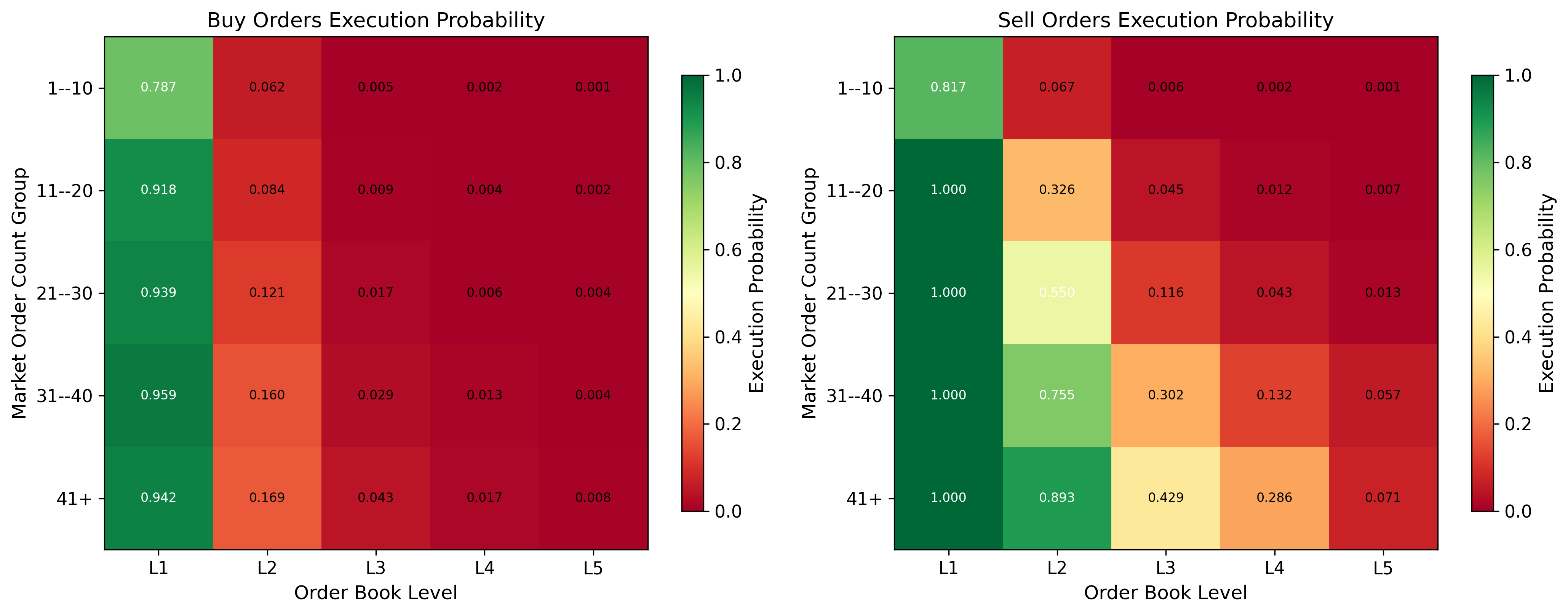}
        \caption{Execution probability heatmaps}
        \label{fig:exec_heatmap}
    \end{subfigure}
    \hfill
    \begin{subfigure}[b]{0.48\textwidth}
        \centering
        \includegraphics[width=\textwidth]{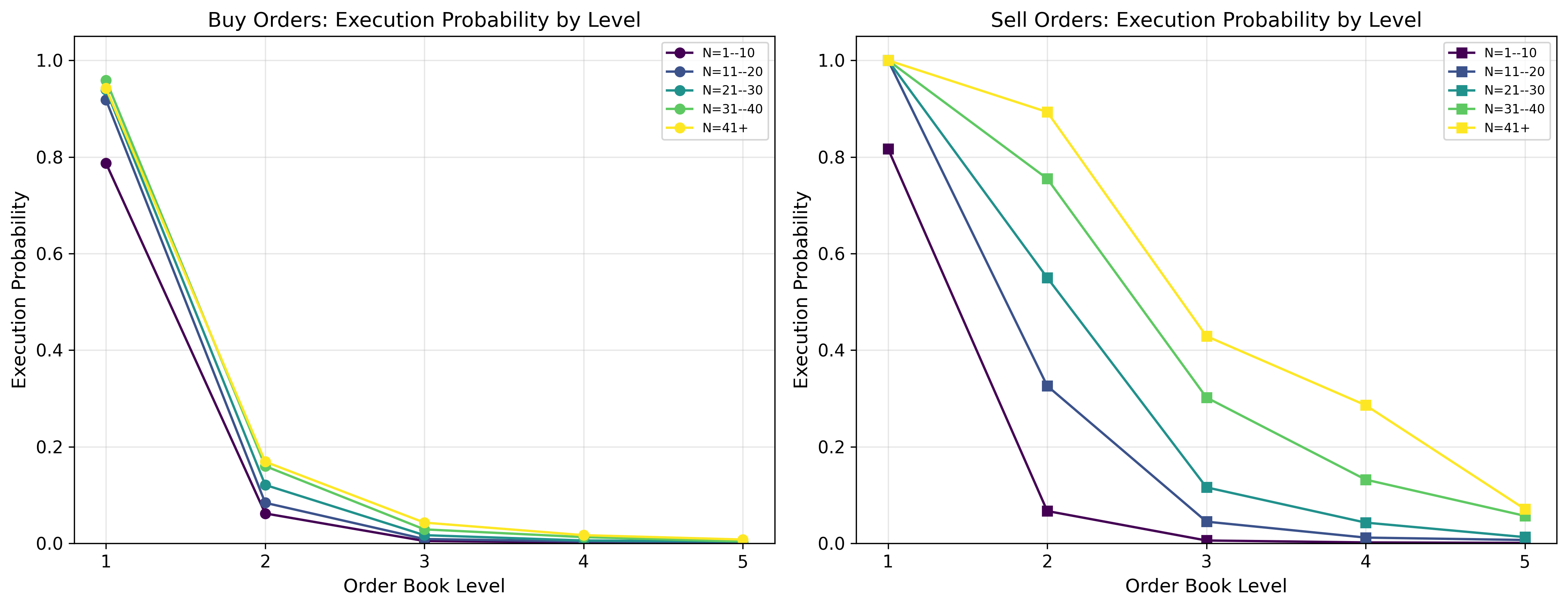}
        \caption{Execution probability curves by group}
        \label{fig:exec_curves}
    \end{subfigure}
    \caption{Execution probability analysis.  (a)~Heatmaps of execution probability by MO count group (rows) and LOB level (columns).  (b)~Execution probability as a function of LOB level for each group.}
    \label{fig:exec_prob_all}
\end{figure}

\subsubsection{Linear Demand Function Estimation}

For each block $j$ and each trading interval, we fit a linear regression of the form $p(\ell) = a + b \times \ell$ to the five execution probabilities (at levels $\ell = 1, \dots, 5$). The regression is effectively an ordinary least-squares fit over the five data points $(\ell, p_\ell)$ for each interval. We collect the fitted intercepts $a_N^\pm$ and slopes $b_N^\pm$ for each block $j$, and compute their means and standard deviations.  Table~\ref{tab:linear_params} reports the results.  Figures~\ref{fig:intercept_comp} and~\ref{fig:slope_comp} visualize the parameter estimates across groups.

\begin{table}[htbp]
\centering
\caption{Linear demand function parameter estimates by MO count group ($\De t = 1$s).  $a_N^\pm$ is the intercept, $b_N^\pm$ is the slope.  Standard deviations in parentheses reflect cross-interval variability within each group.}
\label{tab:linear_params}
\begin{tabular}{@{}lrrrrrr@{}}
\toprule
& \multicolumn{3}{c}{\textbf{Buy Side ($+$)}} & \multicolumn{3}{c}{\textbf{Sell Side ($-$)}} \\
\cmidrule(lr){2-4} \cmidrule(lr){5-7}
Group ($N$) & Intervals & $a_N^+$ & $b_N^+$ & Intervals & $a_N^-$ & $b_N^-$ \\
\midrule
 1--10 & 9,103 &  0.707 (0.260) & $-0.175$ (0.061) & 9,339 &  0.688 (0.268) & $-0.170$ (0.063) \\
11--20 & 1,487 &  0.956 (0.235) & $-0.229$ (0.049) & 1,086 &  0.968 (0.243) & $-0.230$ (0.050) \\
21--30 &   404 &  1.072 (0.260) & $-0.250$ (0.052) &   235 &  1.089 (0.267) & $-0.248$ (0.057) \\
31--40 &   121 &  1.115 (0.261) & $-0.255$ (0.054) &    53 &  1.202 (0.254) & $-0.251$ (0.077) \\
  40+  &    90 &  1.142 (0.267) & $-0.225$ (0.099) &    42 &  1.275 (0.198) & $-0.246$ (0.082) \\
\bottomrule
\end{tabular}
\end{table}

\begin{figure}[htbp]
    \centering
    \begin{subfigure}[b]{0.48\textwidth}
        \centering
        \includegraphics[width=\textwidth]{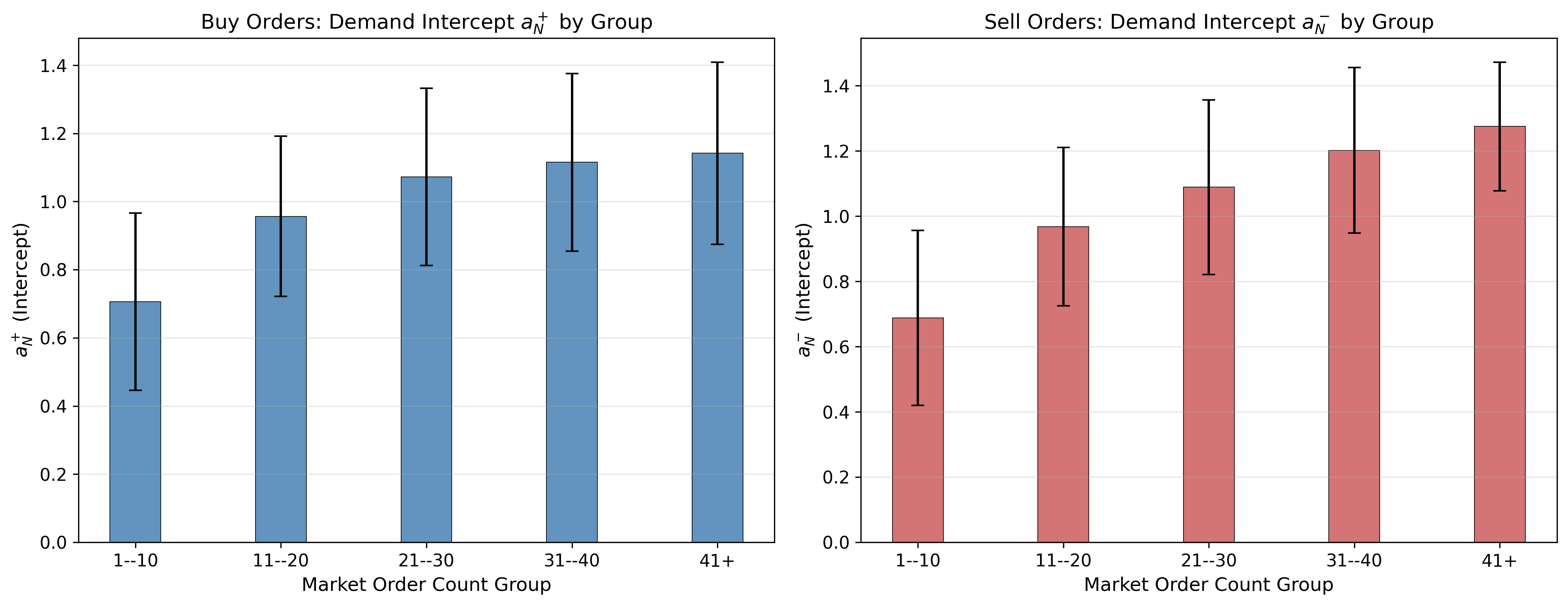}
        \caption{Intercept $a_N^\pm$ by group}
        \label{fig:intercept_comp}
    \end{subfigure}
    \hfill
    \begin{subfigure}[b]{0.48\textwidth}
        \centering
        \includegraphics[width=\textwidth]{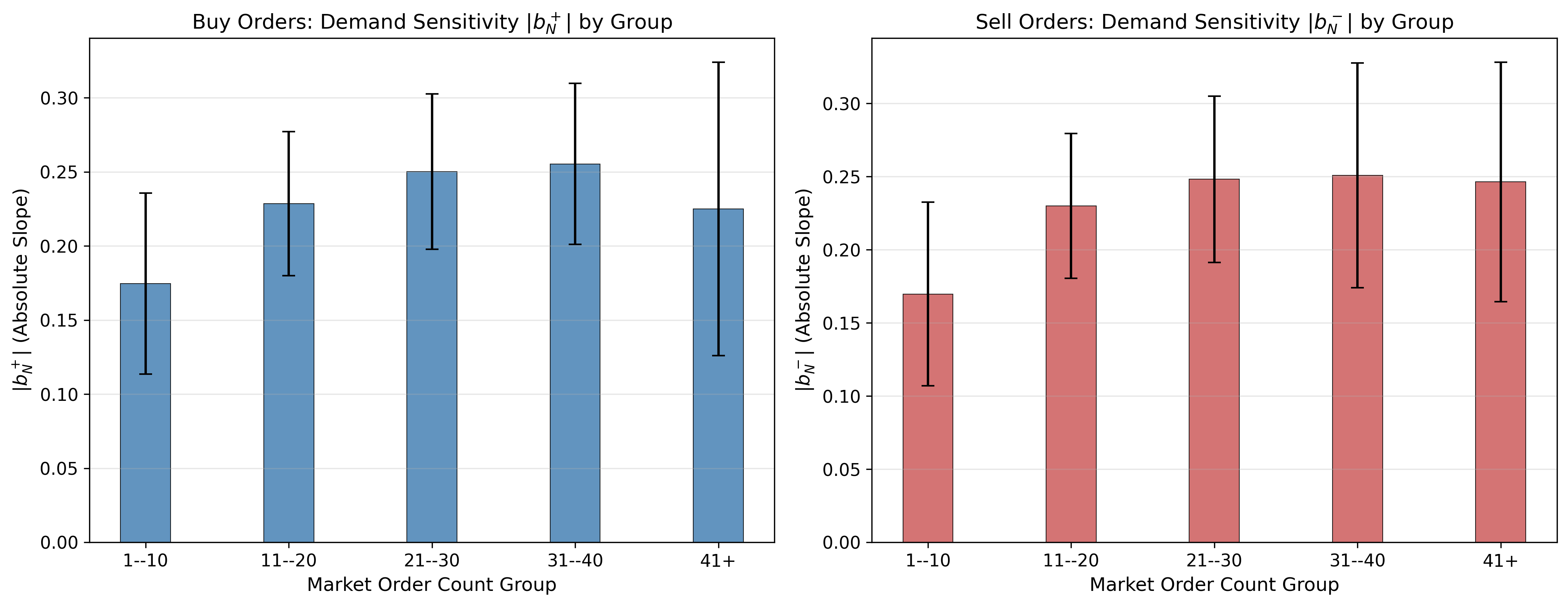}
        \caption{Absolute slope $|b_N^\pm|$ by group}
        \label{fig:slope_comp}
    \end{subfigure}
    \caption{Estimated linear demand function parameters by MO count group.  (a)~Intercept $a_N^\pm$: higher values indicate higher baseline execution probability at Level~1.  (b)~Absolute slope $|b_N^\pm|$: higher values indicate steeper decay of execution probability with level (i.e., more price-sensitive demand).  Error bars denote $\pm 1$ standard deviation.}
    \label{fig:params_all}
\end{figure}

The following patterns emerge from Table~\ref{tab:linear_params}:

\begin{enumerate}
    \item \textbf{Increasing intercept with $N$.}  The intercept $a_N^\pm$ rises monotonically with the MO count.  For buy orders, $a_N^+$ increases from $0.707$ (block $1$--$10$) to $1.142$ (block $40{+}$).  The sell-side intercepts exhibit the same trend and are consistently slightly higher than their buy-side counterparts, consistent with the marginally larger average sell order size.
    \item \textbf{Negative slopes throughout.}  All estimated slopes $b_N^\pm$ are negative, confirming that the execution probability decreases with the level $L$ (i.e., with distance from the mid-price).  This is a standard feature of LOB demand: market participants are less willing to trade at less favorable prices.
    \item \textbf{Increasing sensitivity with $N$.}  The absolute slope $|b_N^\pm|$ generally increases from the $1$--$10$ block to the $31$--$40$ block (e.g., buy side: from $0.175$ to $0.255$; sell side: from $0.170$ to $0.251$), indicating that demand becomes more price-sensitive when more MOs arrive.  In the $40{+}$ block, the slope estimate shows slightly higher variability, likely due to the smaller sample size in that group ($n \le 90$ intervals).
    \item \textbf{Cross-side symmetry.}  The buy-side and sell-side parameter estimates are remarkably similar.  For instance, the buy-side slope for the $1$--$10$ block is $-0.175$ while the sell-side counterpart is $-0.170$.  This symmetry is an important structural property that simplifies the Bellman recursion in Section~4, as the coefficient recursions for $L^+$ and $L^-$ become essentially mirror images.
    \item \textbf{Linear approximation validity.}  The intercept $a_N^\pm$ lies in the range $[0.69, 1.27]$, and the absolute slope $|b_N^\pm|$ lies in $[0.17, 0.26]$.  Since the market maker's optimal quotes typically fall within the first few LOB levels, the linear approximation $f_N^\pm(L) = a_N^\pm + b_N^\pm L$ yields values in $[0, 1]$ for the relevant range of $L$, satisfying the probabilistic interpretation of the demand function.  The standard deviations of the parameter estimates indicate moderate cross-interval variability, which is precisely the stochastic element that the framework of \cite{capponi2021market} and \cite{chavez2024adaptive} is designed to accommodate.
\end{enumerate}





\subsection{Block Probabilities from Calibrated sdHawkes}

The block probabilities
$p_{j}^{\delta} = P\bigl[N^{\delta}(t+\Delta t) - N^{\delta}(t) \in \mathrm{Block}(j)\bigr]$
introduced in Section~2.4 are the only sdHawkes-side inputs consumed by the
demand-function estimation of Section~5.2 and the mixed-moment recursions of
Section~5.4. Both the calibration of the sdHawkes parameters
$(\nu, \alpha, \beta, \phi)$ on tick-level MSFT LOBSTER data and the solution
of the PGF-PDE system for the unconditional per-side marginal are performed in
the companion paper~\cite{research3}, on the same five trading days
(November~3--7, 2014) used throughout this paper. 

The resulting per-side block probabilities at $\Delta t = 1$~s are reported in
Table~\ref{tab:block_probs_r3}. They are marginalised over the stationary
distribution of the queue-imbalance state and are cross-validated against
Ogata-thinning Monte Carlo simulations.

\begin{table}[htbp]
\centering
\caption{Per-side unconditional block probabilities from the calibrated
state-dependent Hawkes model at $\Delta t = 1$~s. Rows sum to unity within
rounding error. Source: companion paper~\cite{research3}.}
\label{tab:block_probs_r3}
\begin{tabular}{lrr}
\hline
$\mathrm{Block}(j)$ & $p_{j}^{+}$ (buy) & $p_{j}^{-}$ (sell) \\
\hline
$0$        & $0.4766$              & $0.5289$              \\
$1$--$10$  & $0.5100$              & $0.4600$              \\
$11$--$20$ & $0.01250$             & $0.01028$             \\
$21$--$30$ & $8.27\times 10^{-4}$  & $7.48\times 10^{-4}$  \\
$31$--$40$ & $7.13\times 10^{-5}$  & $7.32\times 10^{-5}$  \\
$41^{+}$   & $7.99\times 10^{-6}$  & $9.32\times 10^{-6}$  \\
\hline
\end{tabular}
\end{table}

Two features of Table~\ref{tab:block_probs_r3} shape the downstream DPP.
First, the zero-arrival mass is substantial on both sides
($p_{0}^{+} = 0.4766$, $p_{0}^{-} = 0.5289$), so that roughly half of the
one-second intervals contribute no market-order demand; the difference gives
single-side arrival rates $\pi^{+} = 1 - p_{0}^{+} \approx 0.523$ and
$\pi^{-} \approx 0.471$, a $\sim\!11\%$ asymmetry in favour of the buy side
over the sample period. Second, conditional on at least one arrival, the mass
concentrates almost entirely on the first non-zero block:
$p_{1}^{\delta} / (1 - p_{0}^{\delta}) \approx 0.975$ for both sides, with
only $\sim\!2.5\%$ of the conditional mass at $N \geq 11$.

\subsection{Dynamical Programming and Parameters Calibration}


With all model parameters estimated---the linear demand coefficients $a_N^\pm, b_N^\pm$, the sdHawkes parameters $\nu_e, \alpha_{e'xe}, \beta_{e'xe}$, and the closed-form block probabilities $P[N\in\mathrm{Block}(j)]$---we now assemble them into DPP to compute the optimal placement strategy. This section details the calibration pipeline: (i) assembly of per-timestep coefficients, (ii) computation of mixed moments from the block-probability distribution, (iii) backward recursion for the value-function parameters $\alpha_{t_k}, h_{t_k}, g_{t_k}$ and the optimal spread coefficients ${}^{(i)}\!A_{t_k}^{\pm}$, and (iv) forward simulation of the resulting quote and inventory dynamics under practical constraints.


\subsubsection{Calibration Inputs and Demand Moments}
\label{sec:calibration_inputs_moments}

The numerical strategy combines the block-specific demand estimates with the market-order block arrival probabilities.
For each side $\delta\in\{+,-\}$ and nonzero block $j$, let
\[
A_j^\delta=\mathrm{AVOS}\,a_j^\delta,
\qquad
B_j^\delta=\mathrm{AVOS}\,b_j^\delta,
\]
where $a_j^\delta$ and $b_j^\delta$ are the estimated coefficients
reported in Table~\ref{tab:linear_params}. These scaled coefficients
express demand in shares. In the forward evaluation, the realized
market-order count on each side selects the corresponding demand block
for the interval; intervals with no arrivals generate no demand on
that side.

The backward recursion uses demand moments obtained by averaging over
the nonzero blocks. With $p_j^\delta$ denoting the unconditional
block probabilities in Table~\ref{tab:block_probs_r3}, the weights
conditional on at least one arrival are
\begin{equation}
w_j^\delta=\frac{p_j^\delta}{1-p_0^\delta},
\qquad j=1,\ldots,5,\quad \delta\in\{+,-\}.
\label{eq:conditional_weight}
\end{equation}
The weights are computed separately for the two sides and held fixed
throughout the numerical experiment. The resulting demand moments are
\begin{equation}
\begin{aligned}
\bar{\mu}_A^\delta
  &=\sum_{j=1}^{5}w_j^\delta A_j^\delta,
&\qquad
\bar{\mu}_B^\delta
  &=\sum_{j=1}^{5}w_j^\delta B_j^\delta,\\[4pt]
\bar{\mu}_{A^2}^\delta
  &=\sum_{j=1}^{5}w_j^\delta(A_j^\delta)^2,
&
\bar{\mu}_{B^2}^\delta
  &=\sum_{j=1}^{5}w_j^\delta(B_j^\delta)^2,\\[4pt]
\bar{\mu}_{AB}^\delta
  &=\sum_{j=1}^{5}w_j^\delta A_j^\delta B_j^\delta.
\end{aligned}
\label{eq:per_side_moments}
\end{equation}

Since the supplied block probabilities are per-side marginals, they
do not determine the cross-side demand moments. We use the product
approximation
\begin{equation}
\begin{aligned}
\bar{\mu}_{B^+B^-}&=\bar{\mu}_B^+\bar{\mu}_B^-,\\
\bar{\mu}_{B^+A^-}&=\bar{\mu}_B^+\bar{\mu}_A^-,\\
\bar{\mu}_{A^+B^-}&=\bar{\mu}_A^+\bar{\mu}_B^-.
\end{aligned}
\label{eq:cross_moments}
\end{equation}
These factorizations require independence under the demand
distribution being averaged and are an additional modeling assumption,
not a consequence of the Hawkes specification. The calibrated moments
are reported in Table~\ref{tab:mixed_moments_r3}.

\begin{table}[htbp]
\centering
\caption{Calibrated demand moments used in the backward recursion.}
\label{tab:mixed_moments_r3}
\small
\begin{tabular}{lrr}
\toprule
Moment & Buy ($\delta = +$) & Sell ($\delta = -$) \\
\midrule
$\bar{\mu}_{A}^{\delta}$    & $145.22$        & $143.73$        \\
$\bar{\mu}_{B}^{\delta}$    & $-35.85$        & $-35.43$        \\
$\bar{\mu}_{A^{2}}^{\delta}$ & $21{,}157.31$  & $20{,}743.42$   \\
$\bar{\mu}_{B^{2}}^{\delta}$ & $1{,}288.11$   & $1{,}259.39$    \\
$\bar{\mu}_{AB}^{\delta}$    & $-5{,}220.28$  & $-5{,}110.93$   \\
\midrule
\multicolumn{3}{l}{\emph{Cross moments:}} \\
$\bar{\mu}_{B^{+} B^{-}}$ & \multicolumn{2}{c}{$1{,}270.17$}  \\
$\bar{\mu}_{B^{+} A^{-}}$ & \multicolumn{2}{c}{$-5{,}152.16$} \\
$\bar{\mu}_{A^{+} B^{-}}$ & \multicolumn{2}{c}{$-5{,}145.69$} \\
\bottomrule
\end{tabular}
\end{table}

The conditional weight on the first nonzero block is approximately
$97.4\%$ on the buy side and $97.6\%$ on the sell side. Thus, most
of the weight in the demand calibration is assigned to intervals
containing between one and ten market orders. The remaining blocks
retain their contributions through \eqref{eq:per_side_moments}.

In addition to these fixed demand moments, the recursion uses the
time-varying arrival probabilities $\pi_{t_{k+1}}^\pm$ and
$\pi_{t_{k+1}}(1,1)$. In the numerical implementation, these are
estimated from recent arrival activity, represented by
$(g_0,g_1,g_2)$, where
\[
g_m=
\mathbf{1}_{\{N_{\mathrm{buy}}(t_{k-m})>0\}}
+\mathbf{1}_{\{N_{\mathrm{sell}}(t_{k-m})>0\}},
\qquad m=0,1,2.
\]
Here $N_{\mathrm{buy}}(t_{k-m})$ and
$N_{\mathrm{sell}}(t_{k-m})$ denote the interval counts used to
construct the arrival-history categories. For an admissible quote
decision, these counts must be observed by $t_k$. The conditional
arrival probabilities are separate inputs from the fixed marginal
weights in \eqref{eq:conditional_weight}; in particular,
$\pi_{t_{k+1}}(1,1)$ is supplied by the joint arrival-frequency
estimate rather than inferred from the per-side block probabilities.
At the first two decision times of each trading day, the required
history is unavailable, and all three arrival probabilities are set
to zero in the backward recursion. This initialization convention is
applied consistently across the sample.


\subsubsection{Backward Dynamic Programming Recursion}
\label{sec:backward_calibration}

Given the calibrated demand moments and arrival probabilities, we
evaluate the backward recursions from the terminal
decision interval to the beginning of the trading horizon. The
terminal inventory penalty is set to $\lambda=1.0$, with
\begin{equation}
\alpha_T=-\lambda=-1.0,
\qquad h_T=0,
\qquad g_T=0.
\label{eq:calibration_terminal_conditions}
\end{equation}
At each decision time, the recursion determines the value-function
coefficients $\alpha_{t_k}$, $h_{t_k}$, and $g_{t_k}$ together with
the coefficients ${}^{(i)}\!A_{t_k}^{\pm}$, $i=1,2,3$, of the
inventory-dependent quotes. The resulting quote coefficients are
then used in the forward evaluation of the strategy.

The numerical calculation includes a prescribed adjustment of the
denominator $\gamma_{t_k}$. If its computed value is nonnegative,
it is replaced by $-|\gamma_{t_k}|-10^{-12}$ before evaluating
the quote coefficients. This adjustment specifies how the
implementation handles such cases; it does not, by itself,
establish the concavity conditions required by the theoretical
optimization problem. The initialization convention sets
$\pi_{t_{k+1}}^{+}=\pi_{t_{k+1}}^{-}
=\pi_{t_{k+1}}(1,1)=0$ at the first two decision times of each
trading day, when the required arrival history is unavailable.
At these steps, the algebraic prescription returns zero quote
coefficients. These values reflect the initialization convention,
rather than an economically identified optimal spread.

The initial value-function coefficients are
\[
\alpha_0\approx-1.0\times10^{-6},
\qquad h_0\approx0.1041,
\qquad g_0\approx3.210\times10^{6},
\]
consistent with the calibration summary in
Table~\ref{tab:calibration_summary}. These are coefficients of the
candidate value function, rather than realized trading outcomes.
In particular, $g_0$ is the candidate value at zero initial cash
and inventory; realized terminal wealth is evaluated separately
in Section~5.5.

For the same calibration, the approximately constant interior values
of the baseline quote terms are
\[
{}^{(3)}\!A_{t_k}^{+}\approx2.026,
\qquad
{}^{(3)}\!A_{t_k}^{-}\approx2.028,
\]
in tick units, as reported in Section~5.5.1. These terms describe
one component of the unconstrained quotes. Actual quote placement
also depends on the inventory and linear value-function terms,
followed by the inventory adjustments and discretization described
in the next subsection. The resulting quote distribution and
inventory performance are examined in Section~5.5.


\subsubsection{Quote Implementation and Inventory Control}
\label{sec:quote_implementation}

The forward experiment evaluates the inventory-dependent quotes
derived in Section~3.4 using the calibrated coefficients from the
backward recursion. Denote these unconstrained offsets by
$L_{t_k}^{\pm,\mathrm{raw}}$. Before applying the demand model, we
adjust the offsets for inventory exposure and restrict their values
to the five placement levels considered in the experiment. These
adjustments define the implemented strategy; its optimality does
not follow directly from the verification result for the original
control problem.

Inventory exposure is moderated by a soft threshold
$I_{\mathrm{lim}}=5{,}000$ shares. Define the relative excess by
\[
e_{t_k}=\max\!\left\{
\frac{|I_{t_k}|-I_{\mathrm{lim}}}{I_{\mathrm{lim}}},0
\right\},
\qquad \kappa=3.
\]
When inventory exceeds the positive threshold, the ask offset is
compressed and the bid offset is widened; the roles are reversed
when inventory falls below the negative threshold. Specifically,
\begin{equation}
\begin{aligned}
L_{t_k}^{+,\mathrm{soft}}
&=
\begin{cases}
\max\{L_{t_k}^{+,\mathrm{raw}}e^{-\kappa e_{t_k}},0\},
  & I_{t_k}>I_{\mathrm{lim}},\\
L_{t_k}^{+,\mathrm{raw}}(1+\kappa e_{t_k}),
  & I_{t_k}<-I_{\mathrm{lim}},\\
L_{t_k}^{+,\mathrm{raw}},
  & |I_{t_k}|\leq I_{\mathrm{lim}},
\end{cases}
\\[6pt]
L_{t_k}^{-,\mathrm{soft}}
&=
\begin{cases}
L_{t_k}^{-,\mathrm{raw}}(1+\kappa e_{t_k}),
  & I_{t_k}>I_{\mathrm{lim}},\\
\max\{L_{t_k}^{-,\mathrm{raw}}e^{-\kappa e_{t_k}},0\},
  & I_{t_k}<-I_{\mathrm{lim}},\\
L_{t_k}^{-,\mathrm{raw}},
  & |I_{t_k}|\leq I_{\mathrm{lim}}.
\end{cases}
\end{aligned}
\label{eq:soft_inventory_adjustment}
\end{equation}
For nonnegative raw offsets, this adjustment makes the
inventory-reducing side more aggressive and the opposite side less
aggressive. It is a feedback rule rather than a hard inventory
constraint and does not ensure that
$|I_{t_k}|\leq I_{\mathrm{lim}}$.

The adjusted offsets are rounded to the nearest integer and
restricted to the admissible placement levels:
\begin{equation}
L_{t_k}^{\pm,\mathrm{final}}
=\min\!\left\{5,\max\!\left\{1,
\operatorname{round}(L_{t_k}^{\pm,\mathrm{soft}})
\right\}\right\}.
\label{eq:implemented_quote_levels}
\end{equation}
Thus $L_{t_k}^{\pm,\mathrm{final}}\in\{1,2,3,4,5\}$,
matching the level range used to estimate the demand functions.

For each interval, the realized order-count block supplies the
coefficients used to evaluate linear demand at the selected level:
\begin{equation}
Q_{t_{k+1}}^\pm
=A_{t_{k+1}}^\pm
+B_{t_{k+1}}^\pm L_{t_k}^{\pm,\mathrm{final}}.
\label{eq:implemented_linear_demand}
\end{equation}
Cash and inventory are then updated according to the dynamics
in Section~3.2. The quote-level restriction in
\eqref{eq:implemented_quote_levels} does not itself impose the
demand admissibility condition of Section~2.3. In particular,
interpreting \eqref{eq:implemented_linear_demand} as a filled
volume requires its value to lie between zero and
$\mathrm{AVOS}$. The specification above retains the linear
demand calculation without an additional truncation of volumes.

Section~5.5 evaluates the resulting strategy through its quote
distribution, inventory exposure, and terminal wealth relative
to the fixed-spread benchmarks.

\subsection{Optimal Placement Strategy}

%

With the backward DPP coefficients ${}^{(i)}\!A_{t_k}^{\pm}$ fully calibrated, we
now evaluate the resulting optimal placement strategy across all five trading days.
This section presents the complete strategy evaluation: the distribution of posted quote
levels, the intraday quote and inventory dynamics on a representative day, and a
performance comparison against fixed-spread benchmark strategies.


\subsubsection{Optimal Quote Level Distribution}

Figure~\ref{fig:level_distribution} and Table~\ref{tab:level_distribution}
summarise the implemented quote levels over the $117{,}000$ one-second
intervals in the five-day sample, after the inventory adjustment and
rounding and clipping described above. Level~2 accounts for $71.52\%$ of
ask quotes and $77.33\%$ of bid quotes, while levels~4 and~5 together
account for less than $1\%$ on either side. This concentration near
level~2 is consistent with the baseline quote coefficients
${}^{(3)}\!A_{t_k}^{+} \approx 2.026$ and
${}^{(3)}\!A_{t_k}^{-} \approx 2.028$ reported above, although the
implemented quotes also depend on inventory and the subsequent quote
adjustments.

\begin{figure}[htbp]
\centering
\includegraphics[width=\textwidth]{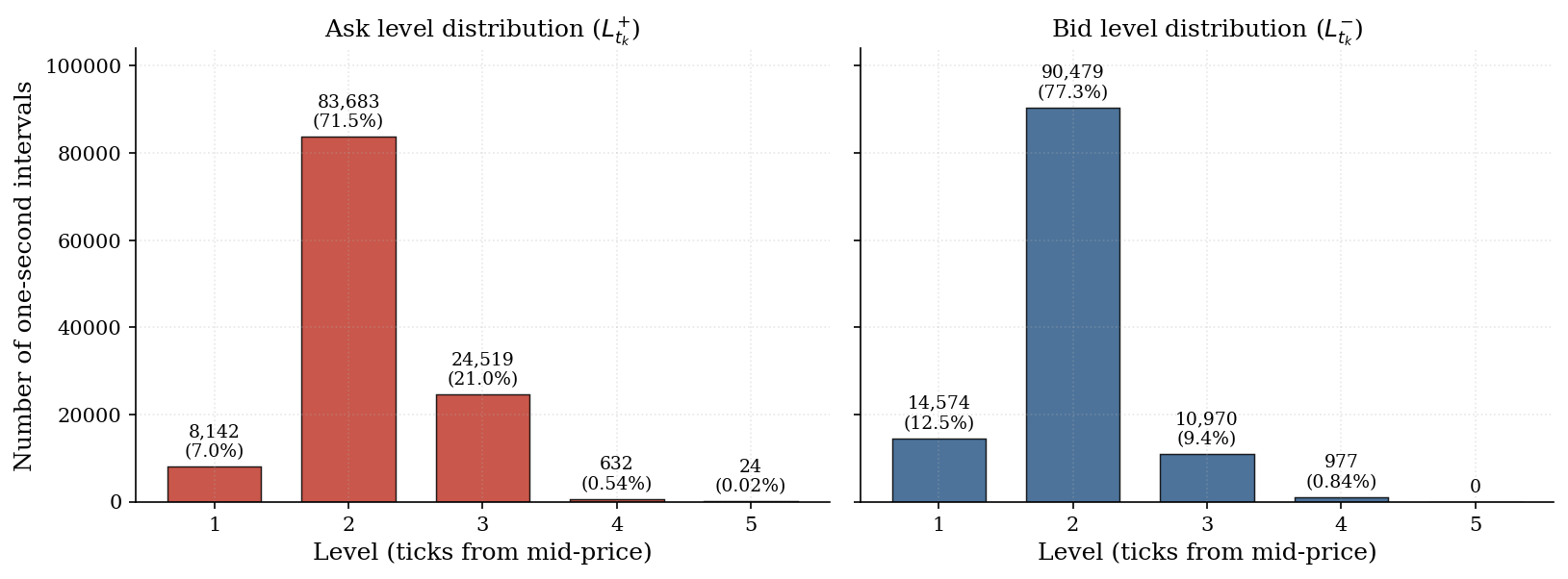}
\caption{Empirical distributions of implemented ask (left) and bid
(right) quote levels over the five-day sample. Quote levels are measured
in ticks from the mid-price.}
\label{fig:level_distribution}
\end{figure}

\begin{table}[htbp]
\centering
\caption{Frequencies of implemented ask and bid quote levels over the five-day sample. Each level is measured in ticks from the mid-price.}\label{tab:level_distribution}
\small
\begin{tabular}{ccccc}
\toprule
Level & Ask steps & Ask (\%) & Bid steps & Bid (\%) \\
\midrule
1 & \phantom{0}8,142  & \phantom{0}6.96 & 14,574 & 12.46 \\
2 & 83,683            & 71.52           & 90,479 & 77.33 \\
3 & 24,519            & 20.96           & 10,970 & \phantom{0}9.38 \\
4 & \phantom{00,0}632 & \phantom{0}0.54 & \phantom{00,0}977 & \phantom{0}0.84 \\
5 & \phantom{000,0}24 & \phantom{0}0.02 & \phantom{000,00}0 & \phantom{0}0.00 \\
\midrule
Total & 117,000 & 100.00 & 117,000 & 100.00 \\
\bottomrule
\end{tabular}
\end{table}

The distributions also exhibit a directional asymmetry: ask quotes occur
more frequently at level~3 than bid quotes ($20.96\%$ versus $9.38\%$),
whereas bid quotes occur more frequently at level~1 ($12.46\%$ versus
$6.96\%$). Thus, relative to the bid distribution, the ask distribution
places more mass at larger offsets. This pattern is consistent with
inventory management in a sample with a negative mean inventory of
approximately $-2{,}470$ shares: wider ask offsets and narrower bid
offsets tend to favour inventory replenishment. The aggregate
frequencies alone, however, do not isolate the contributions of order-flow
asymmetry, inventory feedback, and the soft inventory adjustment.

Figure~\ref{fig:inventory_diagnostic}, discussed in the following
subsection, provides the corresponding inventory evidence. The
intraday analysis there complements the pooled distributions by
examining the evolution of quotes and inventory within a trading day.


\subsubsection{Inventory Exposure and Intraday Quote Dynamics}
\label{sec:inventory_quote_diagnostics}

We evaluate the implemented strategy by examining inventory exposure
across the five trading days and quote behavior on the day with the
most frequent threshold exceedances. Table~\ref{tab:calibration_summary}
collects the calibration settings and aggregate inventory statistics.
Inventory remains within the $\pm5{,}000$-share threshold for
$47.09\%$ of the sample, while the soft adjustment is active for
the remaining $52.91\%$. The maximum absolute inventory is
$6{,}311.2$ shares, approximately $26.2\%$ above the threshold.
These results confirm that the threshold acts as a soft adjustment
trigger rather than a binding inventory limit. The observed maximum
is a sample statistic, not a guaranteed bound on future exposure.

\begin{table}[htbp]
\centering
\caption{Calibration settings and aggregate inventory statistics for the five-day sample.}
\label{tab:calibration_summary}
\small
\begin{tabular}{lr}
\toprule
Metric & Value \\
\midrule
Total timesteps ($T$)                & 117,000 \\
Terminal penalty $\lambda$           & 1.0 \\
Soft barrier limit $I_{\text{lim}}$  & 5,000 shares \\
Barrier sharpness $\kappa$           & 3.0 \\
Integer levels allowed               & $\{1,2,3,4,5\}$ \\
\midrule
$\alpha_{0}$ (initial)               & $-1.0 \times 10^{-6}$ \\
$h_{0}$ (initial)                    & $0.1041$ \\
$g_{0}$ (initial)                    & $3.210 \times 10^{6}$ \\
\midrule
Soft barrier activations             & 61,908 / 117,000 (52.91\%) \\
Steps within $\pm 5{,}000$ inventory & 55,092 / 117,000 (47.09\%) \\
Max $|I_{t_{k}}|$                    & 6,311.2 shares \\
Mean $I_{t_{k}}$                     & $-2{,}469.5$ shares \\
Std $I_{t_{k}}$                      & 4,039.6 shares \\
\bottomrule
\end{tabular}
\end{table}

The mean inventory is $-2{,}469.5$ shares, with a standard deviation
of $4{,}039.6$ shares. The aggregate mean masks substantial variation
across sessions. Table~\ref{tab:per_day_inventory} reports the daily
statistics, and Figure~\ref{fig:inventory_diagnostic} displays the
inventory path and the corresponding adjustment indicator. November~5
has the largest within-threshold proportion ($67.4\%$), whereas
November~7 has the smallest ($34.2\%$). Thus, inventory control
remains an important limitation of the implemented strategy even
with the additional feedback rule.

\begin{table}[htbp]
\centering
\caption{Daily inventory statistics for the implemented strategy. The threshold is $5{,}000$ shares in absolute value.}
\label{tab:per_day_inventory}
\small
\begin{tabular}{cccc}
\toprule
Date & $\max|I_{t_{k}}|$ & Within-limit \% & Barrier steps \\
\midrule
2014-11-03 & 6,217 & 47.8\% & 12,204 \\
2014-11-04 & 6,188 & 37.8\% & 14,564 \\
2014-11-05 & 5,724 & 67.4\% & \phantom{0}7,634 \\
2014-11-06 & 6,251 & 48.2\% & 12,112 \\
2014-11-07 & 6,311 & 34.2\% & 15,394 \\
\midrule
Total      & 6,311 & 47.09\% & 61,908 \\
\bottomrule
\end{tabular}
\end{table}

\begin{figure}[htbp]
\centering
\includegraphics[width=\textwidth]{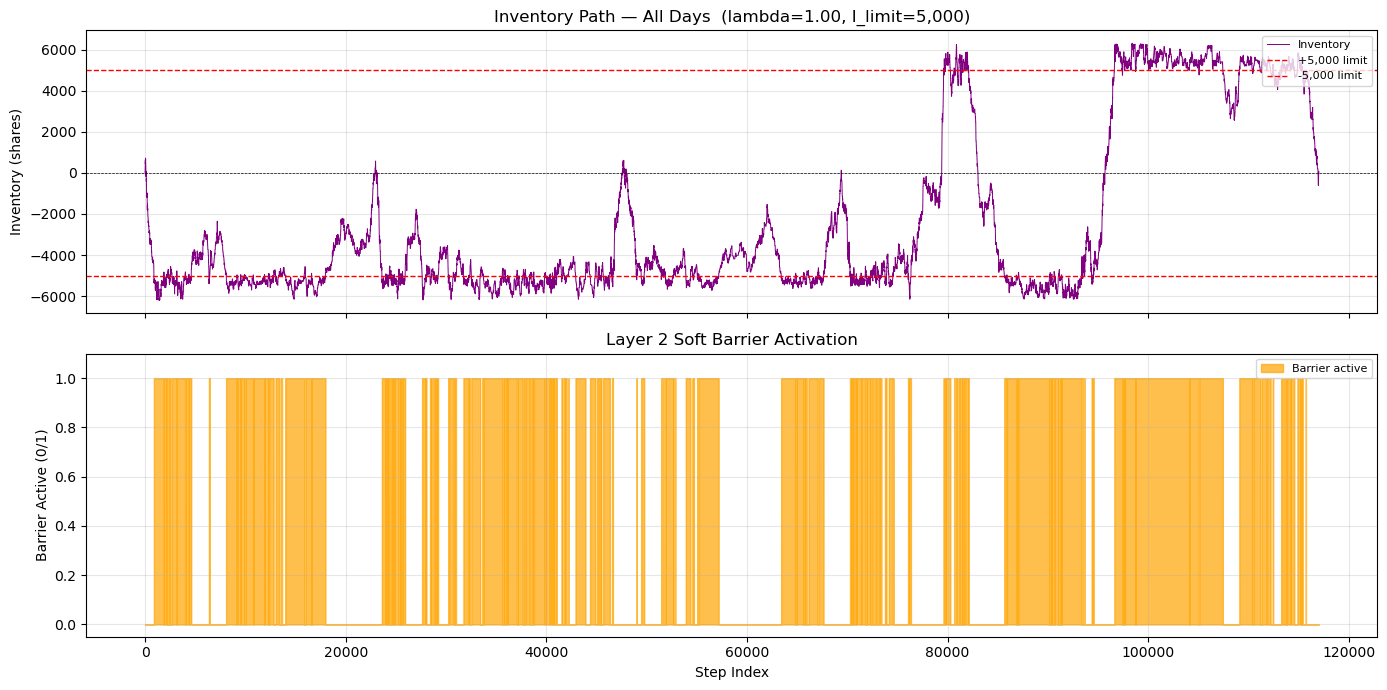}
\caption{Inventory exposure over the five-day sample. Top: inventory and the $\pm5{,}000$-share thresholds. Bottom: the soft-adjustment activation indicator.}
\label{fig:inventory_diagnostic}
\end{figure}

November~7 provides a useful case study of the strategy under sustained
inventory pressure. On that day, the adjustment is active in
$15{,}394$ of $23{,}400$ intervals, and the maximum absolute inventory
is approximately $6{,}311$ shares. The terminal position before
liquidation is $4{,}193$ shares. Figure~\ref{fig:forward_simulation}
shows the joint evolution of quotes, inventory, and cash, while
Table~\ref{tab:trading_summary_1107} collects the single-day results.

\begin{figure}[htbp]
\centering
\includegraphics[width=\textwidth]{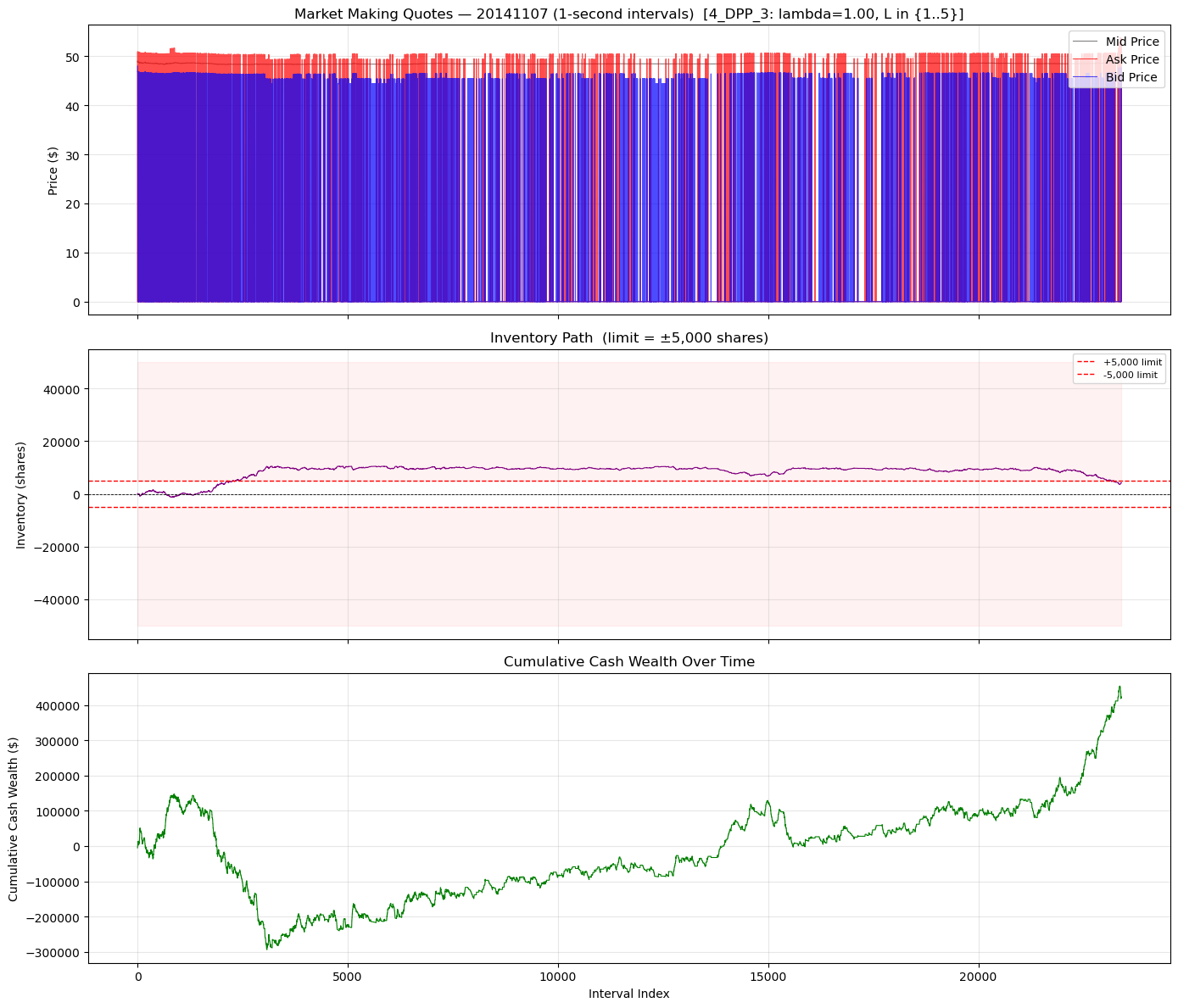}
\caption{Implemented strategy on November 7, 2014. Top: ask and bid quotes and the mid-price. Middle: inventory and the $\pm5{,}000$-share thresholds. Bottom: cumulative cash before terminal liquidation.}
\label{fig:forward_simulation}
\end{figure}

\begin{table}[htbp]
\centering
\caption{Single-day results and quote-level counts for November 7, 2014. Terminal wealth includes liquidation at the closing mid-price.}
\label{tab:trading_summary_1107}
\small
\begin{tabular}{lr}
\toprule
Metric & Value \\
\midrule
Total intervals simulated                          & 23,400 \\
MSFT closing mid-price                             & \$48.685 \\
Terminal inventory (before liquidation)            & $+4{,}193$ shares \\
Terminal inventory (after liquidation)             & 0 \\
Terminal cash wealth $W_{T}$ (incl.\ liquidation)  & \$627,259 \\
\midrule
Max $|I_{t_{k}}|$ on this day                       & 6,311 shares \\
Steps within $\pm 5{,}000$ inventory                & \phantom{00,}8,006 / 23,400 (34.2\%) \\
Soft-barrier activations                           & 15,394 / 23,400 (65.8\%) \\
\midrule
\multicolumn{2}{l}{\emph{Ask level counts on this day:}} \\
$L^{+}=1$ / $2$ / $3$ / $4$ / $5$
    & 7{,}708 / 15{,}580 / 69 / 19 / 24 \\
\multicolumn{2}{l}{\emph{Bid level counts on this day:}} \\
$L^{-}=1$ / $2$ / $3$ / $4$ / $5$
    & 52 / 12{,}154 / 10{,}227 / 967 / 0 \\
\bottomrule
\end{tabular}
\end{table}

The quote-level distribution on this long-inventory day differs from
the five-day aggregate. Ask quotes are concentrated at levels~1 and~2,
whereas bid quotes are concentrated at levels~2 and~3. This direction
is consistent with more aggressive selling and less aggressive buying
when inventory is positive. Figure~\ref{fig:quote_levels_1107}
compares the adjusted continuous offsets with the final integer levels,
illustrating how discretization translates the inventory response
into changes in posted levels. The aggregate distribution is reported
in Table~\ref{tab:level_distribution}.

\begin{figure}[htbp]
\centering
\includegraphics[width=\textwidth]{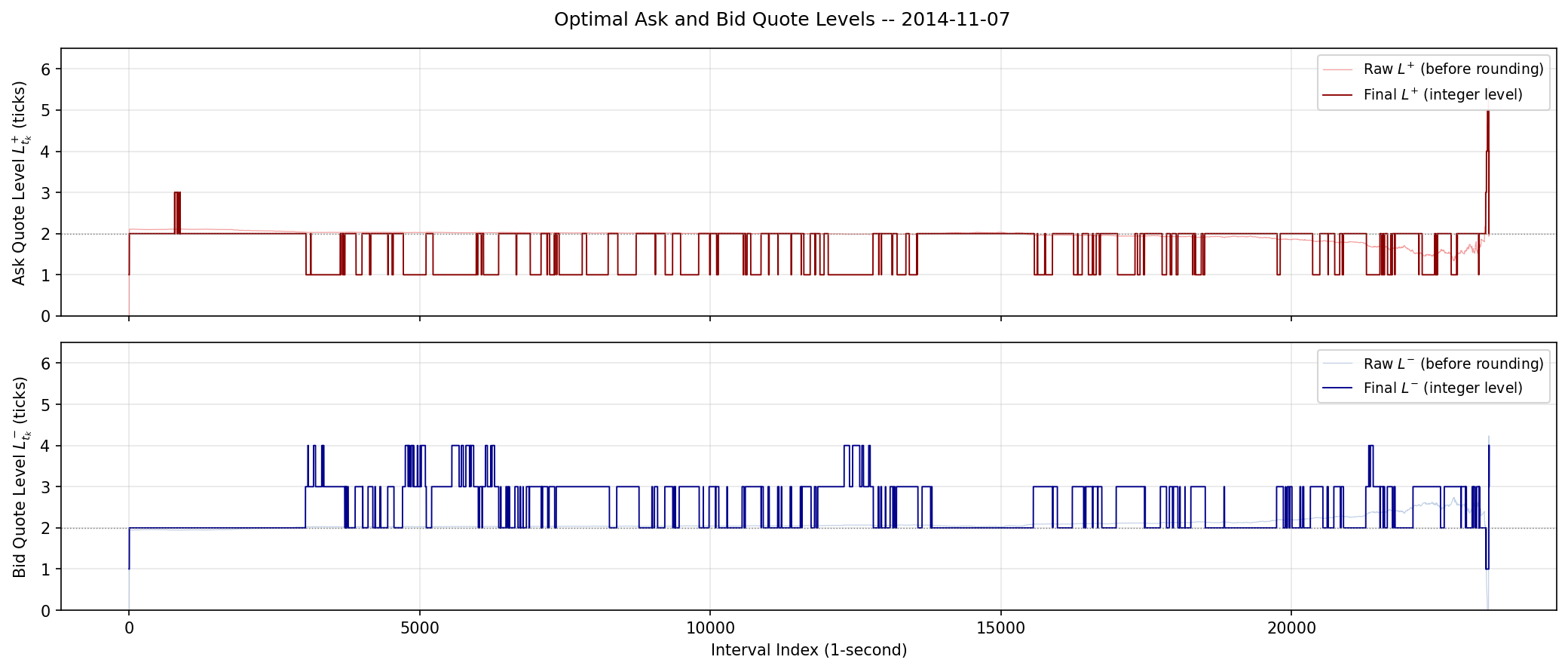}
\caption{Ask (top) and bid (bottom) placement levels on November 7, 2014. Light curves show the adjusted continuous offsets $L_{t_k}^{\pm,\mathrm{soft}}$; dark curves show the final integer levels. The dotted reference line is at level~2.}
\label{fig:quote_levels_1107}
\end{figure}

The terminal wealth reported for November~7 is $\$627{,}259$ after
liquidation at the closing mid-price. The cash trajectory in
Figure~\ref{fig:forward_simulation} excludes the contemporaneous value
of inventory and should therefore be distinguished from
marked-to-market wealth. Multi-day wealth and the comparison with
fixed-spread strategies are examined in the following subsections.


\subsubsection{Multi-Day Performance}

Table~\ref{tab:multiday_performance} reports the per-day terminal wealth
$W_{T}$ (including inventory liquidation at the closing mid-price) across
all five trading days in November~2014.
Figure~\ref{fig:cumulative_pnl} plots the pre-liquidation cumulative cash
wealth trajectory on each day.

\begin{table}[htbp]
\centering
\caption{Per-day terminal wealth for the optimal DPP strategy
($\lambda = 1.0$, $L_{t_{k}}^{\pm} \in \{1,\dots,5\}$,
$I_{\text{lim}} = 5{,}000$~shares).}
\label{tab:multiday_performance}
\small
\begin{tabular}{lccc}
\toprule
Date & Closing Mid-Price & $W_{T}$ (incl.\ liquidation) & Day's Return \\
\midrule
2014-11-03 & \$47.445 & \$\phantom{0}563{,}110 & Baseline \\
2014-11-04 & \$47.565 & \$\phantom{0}597{,}349 & $+$\,6.1\% \\
2014-11-05 & \$47.855 & \$\phantom{0}600{,}556 & $+$\,6.6\% \\
2014-11-06 & \$48.685 & \$\phantom{0}776{,}807 & $+$37.9\% \\
2014-11-07 & \$48.685 & \$\phantom{0}627{,}259 & $+$11.4\% \\
\midrule
\textbf{Total (5 days)} & --- & \textbf{\$3{,}165{,}081} & --- \\
\bottomrule
\end{tabular}
\end{table}

\begin{figure}[htbp]
\centering
\includegraphics[width=\textwidth]{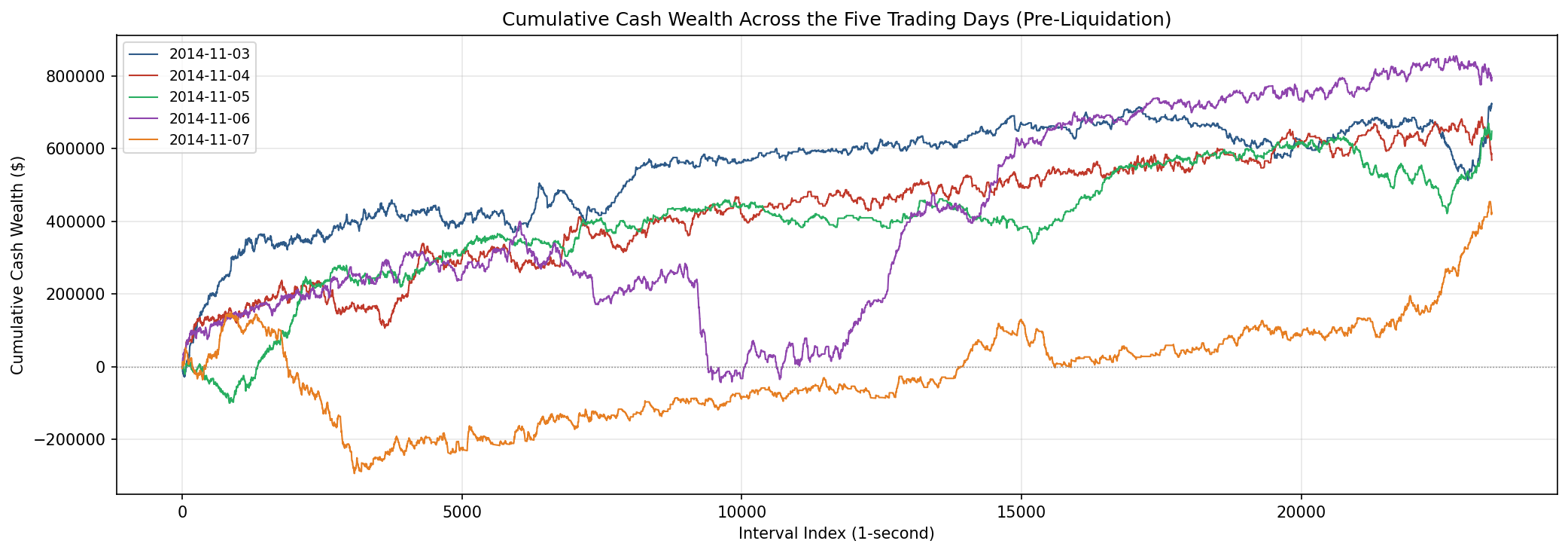}
\caption{Pre-liquidation cumulative cash wealth on each of the five
trading days. Each curve starts from zero and ends at the pre-liquidation
wealth level; the terminal-inventory liquidation contributes the unshown
increment to $W_{T}$ reported in Table~\ref{tab:multiday_performance}.
November~6 exhibits the highest intraday P\&L accumulation, consistent
with its higher market-order volume ($33{,}574$ total orders vs.\
$22{,}566$--$26{,}977$ for the other days).}
\label{fig:cumulative_pnl}
\end{figure}

\noindent
The strategy generates positive terminal wealth on every day, ranging
from \$563{,}110 (November~3) to \$776{,}807 (November~6), for an
aggregate of \$3.17~million over the five-day window. November~6 stands
out as the most profitable day: its unusually high market-order arrival
rate ($33{,}574$ total orders, compared to $22{,}566$--$26{,}977$ on the
other days) translates into higher trade volumes at the posted quotes,
amplifying the per-interval spread capture. The relatively stable P\&L
on the remaining four days reflects the consistency of the
demand-function estimates and the sdHawkes arrival probabilities across
typical market conditions, even as the intraday inventory path itself
varies substantially between short-biased sessions (November~3--5) and
long-biased sessions (November~6--7).


\subsubsection{Comparison with Fixed-Spread Benchmarks}

We compare the implemented DPP strategy with three benchmarks that post
constant ask and bid offsets of $L\in\{1,2,3\}$ ticks. All strategies
are evaluated using the same observed market-order arrival sequence and
the same demand-function coefficients $A_j^{\pm}$ and $B_j^{\pm}$.
The comparison therefore measures the realised performance of adaptive
and fixed quotes under a common demand specification over the five-day
sample.

Table~\ref{tab:benchmark_comparison} reports aggregate terminal wealth
after liquidating each day's remaining inventory at the closing
mid-price. These amounts include the terminal inventory valuation but
exclude the quadratic inventory penalty in the optimisation objective.
The percentage differences use the DPP strategy's aggregate wealth as
the denominator. Figure~\ref{fig:wealth_comparison} presents the
corresponding daily results.

\begin{table}[htbp]
\centering
\caption{Aggregate terminal wealth after daily liquidation at the closing
mid-price. Percentage differences are relative to the implemented DPP
strategy.}
\label{tab:benchmark_comparison}
\small
\begin{tabular}{lcccc}
\toprule
Strategy & Ask level & Bid level & Terminal wealth & Difference (\%) \\
\midrule
DPP strategy & Adaptive & Adaptive & \$3{,}165{,}081 & --- \\
Fixed $L=1$ & 1 & 1 & \$2{,}549{,}852 & $-19.4$ \\
Fixed $L=2$ & 2 & 2 & \$3{,}443{,}368 & $+8.8$ \\
Fixed $L=3$ & 3 & 3 & \$2{,}695{,}842 & $-14.8$ \\
\bottomrule
\end{tabular}
\end{table}

\begin{figure}[htbp]
\centering
\includegraphics[width=\textwidth]{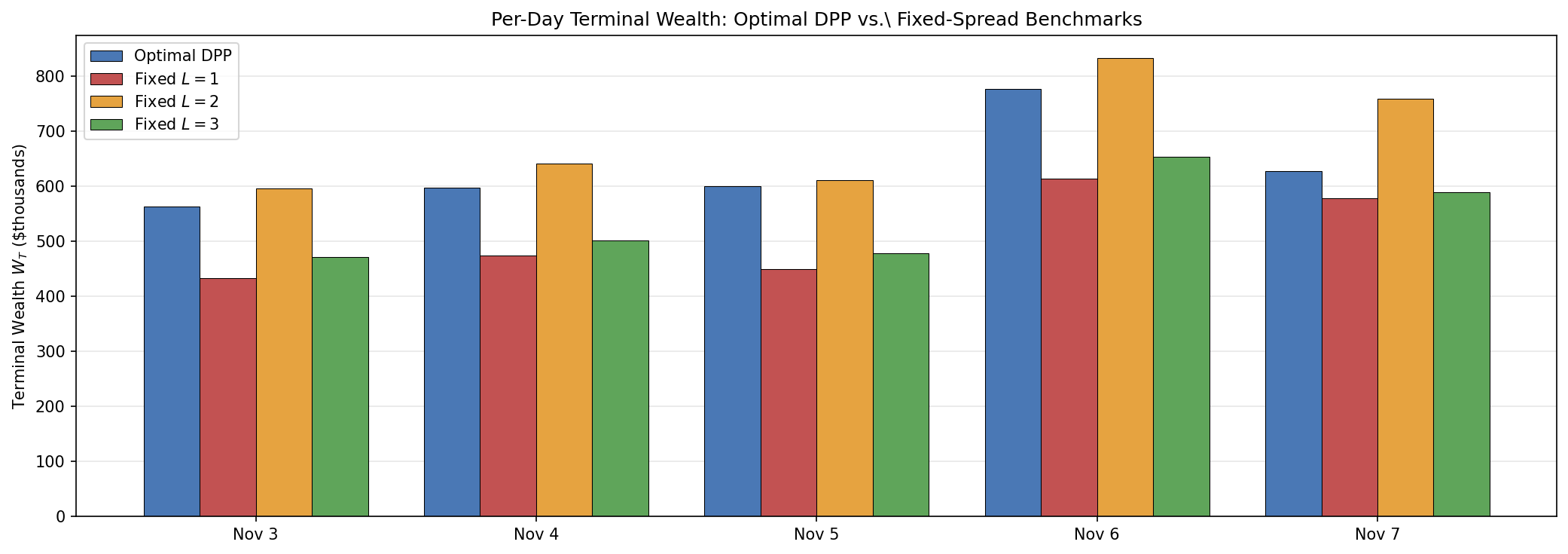}
\caption{Daily terminal wealth after liquidation at the closing
mid-price for the implemented DPP strategy and the three fixed-spread
benchmarks.}
\label{fig:wealth_comparison}
\end{figure}

The DPP strategy produces greater aggregate terminal wealth than the
fixed-$L=1$ and fixed-$L=3$ benchmarks, but falls below fixed $L=2$ by
\$278{,}287. Thus, among the strategies considered, fixed $L=2$ achieves
the highest aggregate wealth in this sample. Its performance is
consistent with the spread--volume trade-off represented by the demand
model: a larger offset increases the spread earned per executed share
but reduces model-implied demand when $B_j^{\pm}<0$. The concentration
of DPP quotes at level~2 in Table~\ref{tab:level_distribution} places
the adaptive strategy near this benchmark for much of the sample,
while its inventory-dependent deviations yield different trading and
inventory paths.

This wealth ranking must be distinguished from a comparison of the
control objective $\mathbb{E}[W_T+S_T I_T-\lambda I_T^2]$.
To make the accounting explicit, let $W_T^{\mathrm{pre}}$ and
$I_T^{\mathrm{pre}}$ denote cash and inventory immediately before
terminal liquidation. Liquidation at the closing mid-price yields
$W_T^{\mathrm{pre}}+S_T I_T^{\mathrm{pre}}$, whereas evaluation of the
penalised terminal criterion additionally subtracts
$\lambda (I_T^{\mathrm{pre}})^2$. The penalty must therefore be
computed using inventory before liquidation; using the zero inventory
after liquidation would remove the intended penalty. The wealth totals
in Table~\ref{tab:benchmark_comparison} alone do not establish which
strategy achieves the higher penalised payoff. Such a comparison
requires the corresponding terminal inventories for every strategy,
with intraday inventory exposure reported separately.

The verification result in Section~3 concerns the controls and
admissible class specified by the theoretical model. The numerical
strategy additionally applies a soft inventory adjustment, rounding,
and clipping, so its optimality does not follow directly from that
result. Accordingly, the present comparison establishes the observed
wealth ranking under the stated simulation specification; it does not
establish superiority in the penalised objective or a reduction in
inventory risk relative to the fixed-spread benchmarks.

\section{Conclusions and future work}


The numerical study uses five trading days of MSFT LOBSTER data from
November~3--7, 2014. The implemented strategy posts predominantly at the
second quote level and exhibits asymmetric ask and bid distributions
consistent with its inventory-dependent adjustments. It produces
positive terminal wealth after liquidation at the closing mid-price on
each day, with an aggregate of approximately \$3.17~million. Inventory
remains within the $5{,}000$-share soft threshold in absolute value for
$47.09\%$ of the sampled intervals, with a maximum absolute position of
$6{,}311.2$ shares. These observations describe the realised inventory
exposure; the soft adjustment does not enforce a hard inventory bound.
In the benchmark comparison, aggregate terminal wealth exceeds that of
the fixed-$L=1$ and fixed-$L=3$ strategies but is lower than that of
fixed $L=2$, whose aggregate wealth is $8.8\%$ higher. Because these
wealth totals exclude the terminal inventory penalty, they do not
establish a ranking under the penalised control objective.

The numerical findings are subject to several limitations. The sample
covers a single stock over five trading days and therefore provides
limited evidence about performance across market conditions. The linear
demand specification is analytically convenient, but its interpretation
as an execution probability requires an appropriate range of quote
levels and coefficients. In addition, the soft inventory adjustment and
discretisation of quotes modify the controls derived in the theoretical
model, so the verification result does not directly establish optimality
of the implemented strategy. The reported results should consequently
be interpreted as a sample-based evaluation under the specified demand
and liquidation assumptions.

Further work should extend the empirical evaluation across securities
and trading periods and compare the strategies using both penalised
terminal payoffs and intraday inventory exposure. Such comparisons
should use a common accounting convention, with the terminal inventory
penalty evaluated before liquidation. A complementary theoretical
direction is to incorporate discrete quote levels and inventory
restrictions directly into the control problem. Finally, bounded or
nonlinear demand specifications could improve the representation of
execution at larger offsets; their empirical fit and implications for
the tractability of the backward recursion remain to be investigated.

\newpage




\bibliographystyle{apalike}
\bibliography{bibliography}

\newpage





\typeout{=== BIBTEX DEBUG ===}
\nocite{*} 
\typeout{=== END DEBUG ===}
\end{document}